\documentclass[english,journal]{extarticle}
\usepackage{fontenc}
\usepackage[utf8]{inputenc}
\usepackage{color}
\usepackage{subfigure}
\usepackage{graphicx}
\usepackage{babel}
\usepackage{float}
\usepackage{amsthm}
\usepackage{amsmath}
\usepackage{amssymb}
\usepackage[usenames,dvipsnames]{xcolor}
\usepackage{tikz}
\usepackage{fixltx2e}
\usepackage{multirow}
\usepackage{anysize}
\usepackage{parskip}
\usepackage{url}
\usepackage[unicode=true,pdfusetitle,
 bookmarks=true,bookmarksnumbered=false,bookmarksopen=false,
 breaklinks=false,pdfborder={0 0 0},backref=false,colorlinks=true]
 {hyperref}
\marginsize{1.05in}{1.05in}{1.05in}{1.05in}

\hypersetup{
    colorlinks=true,       
    linkcolor=Maroon,          
    citecolor=OliveGreen,        
    filecolor=magenta,      
    urlcolor=Blue           
}

\makeatletter

\providecommand{\tabularnewline}{\\}
\floatstyle{ruled}
\newfloat{algorithm}{tbp}{loa}
\providecommand{\algorithmname}{Algorithm}
\floatname{algorithm}{\protect\algorithmname}

  \theoremstyle{plain}
  \newtheorem{thm}{\protect\theoremname}
  \providecommand{\cnjname}{Conjecture}
  \theoremstyle{plain}
  \newtheorem{cnj}{\protect\cnjname}
  \theoremstyle{plain}
  \newtheorem{lem}{\protect\lemmaname}
  \theoremstyle{definition}
  \newtheorem{defn}{\protect\definitionname}
  \theoremstyle{plain}
  \newtheorem{cor}{\protect\corollaryname}
 \ifx\proof\undefined\
   \newenvironment{proof}[1][\proofname]{\par
     \normalfont\topsep6\p@\@plus6\p@\relax
     \trivlist
     \itemindent\parindent
     \item[\hskip\labelsep
           \scshape
       #1]\ignorespaces
   }{%
     \endtrivlist\@endpefalse
   }
   \providecommand{\proofname}{Proof}
 \fi
  \theoremstyle{plain}
  \newtheorem{prop}{\protect\propositionname}

\usepackage{algpseudocode}
\usepackage{hhline}
\usepackage{array,booktabs}

\renewcommand{\algorithmicrequire}{\textbf{Input:}}
\renewcommand{\algorithmicensure}{\textbf{Output:}}

\AtBeginDocument{
  
}

\makeatother

\providecommand{\corollaryname}{\inputencoding{latin9}Corollary}
\providecommand{\definitionname}{\inputencoding{latin9}Definition}
\providecommand{\lemmaname}{\inputencoding{latin9}Lemma}
\providecommand{\propositionname}{\inputencoding{latin9}Proposition}
\providecommand{\theoremname}{\inputencoding{latin9}Theorem}

\newcommand{\mathref}[2]{\hyperref[msym:#1]{#2}}

\begin{document}

\title{Fast and efficient exact synthesis of single qubit unitaries generated
by Clifford and T gates}

\author{\small{Vadym Kliuchnikov$^{1}$, Dmitri Maslov$^{2, 3}$ and Michele
Mosca$^{4,5}$} \\
{\small\it $^1$ Institute for Quantum Computing, and David R. Cheriton School of
Computer Science} \\
{\small\it University of Waterloo, Waterloo, Ontario, Canada} \\
{\small\it $^2$ National Science Foundation} \\
{\small\it Arlington, Virginia, USA} \\
{\small\it $^3$  Institute for Quantum Computing, and Dept. of Physics \&
Astronomy} \\
{\small\it University of Waterloo,  Waterloo, Ontario, Canada} \\
{\small\it $^4$  Institute for Quantum Computing, and Dept. of Combinatorics \&
Optimization} \\
{\small\it University of Waterloo,  Waterloo, Ontario, Canada} \\
{\small\it $^5$ Perimeter Insitute for Theoretical Physics} \\
{\small\it Waterloo, Ontario, Canada} \\
}

\maketitle

\global\long\def\w{\omega}
\global\long\def\Z{\mathbb{Z}}
\global\long\def\G{\mathcal{G}}
\global\long\def\Zr{\mathbb{Z}\left[\w\right]}
\global\long\def\N{\mathbb{N}}
\global\long\def\ox{\overline{x}}
\global\long\def\oy{\overline{y}}
\global\long\def\gde{\mathbf{\mathrm{gde}}}
\global\long\def\mt#1{\left(\mathrm{mod\,2^{#1}}\right)}
\global\long\def\m{\left(\mathrm{mod\,2}\right)}
\global\long\def\sde{\mathrm{sde}}
\global\long\def\re{\mathrm{Re}}
\global\long\def\Rot{\mathrm{R}}
\global\long\def\by{\times}
\global\long\def\abs#1{\left|#1\right|}
\global\long\def\R{\mathbb{Z}[\frac{1}{\sqrt{2}},i]}
\global\long\def\Zr{\mathbb{Z}\left[\w\right]}
\global\long\def\ip#1#2{\left\langle #1,#2\right\rangle }
\global\long\def\F#1#2{F\left( #1,#2\right) }
\global\long\def\Q#1{Q\left(#1\right) }
\global\long\def\P#1{P\left(#1\right) }
\global\long\def\bra#1{\left\langle #1\right|}
\global\long\def\ket#1{\left|#1\right\rangle }
\global\long\def\set#1#2{\left\{  \left.#1\,\right|\,#2\right\}  }

\newcommand\ST{\rule[-1em]{0pt}{3em}}
\global\long\def\tbl#1{\ST\underset{\,}{\overset{\,}{#1}}}


\floatstyle{ruled}
\newfloat{algorithm}{tbp}{loa}
\providecommand{\algorithmname}{Algorithm}
\floatname{algorithm}{\protect\algorithmname}

\renewcommand{\algorithmicrequire}{\textbf{Input:}}
\renewcommand{\algorithmicensure}{\textbf{Output:}}

\renewcommand{\leftmark}{\footnotesize\it{\quad Fast and efficient exact synthesis of single-qubit unitaries generated by Clifford and T gates}\hfill}
\renewcommand{\rightmark}{\hfill\footnotesize\it{Vadym Kliuchnikov, Dmitri Maslov, Michele Mosca\quad}}

\begin{abstract}
In this paper, we show the equivalence of the set of unitaries computable by the
circuits over the Clifford and T library and the set of unitaries over the ring
$\mathbb{Z}[\frac{1}{\sqrt{2}},i]$, in the single-qubit case.  We report an
efficient synthesis algorithm, with an exact optimality guarantee on the number of
Hadamard and T gates used.  We conjecture that the equivalence of the sets of
unitaries implementable by circuits over the Clifford and T library and
unitaries
over the ring $\mathbb{Z}[\frac{1}{\sqrt{2}},i]$ holds in the $n$-qubit case. 
\end{abstract}

\section{Introduction} 
The problem of efficient approximation of an arbitrary unitary using a finite
gate set is important in quantum computation.  In particular,
fault tolerance methods impose limitations on the set of elementary gates that
may be used on the logical (as opposed to physical) level. 
One of the most common of such sets consists of Clifford\footnote{Also known as stabilizer gates/library.  In the single-qubit case the Clifford library consists of, e.g., Hadamard and Phase gates.  In the multiple qubit case, the two-qubit CNOT gate is also included in the Clifford library.} ~and
T$:=\left(\begin{array}{cc}
1 & 0\\
0 & e^{i\pi/4}
\end{array}\right)$ gates.  This gate library is known to be approximately
universal
in the sense of the existence of an efficient approximation of the unitaries by
circuits over it.
In the single-qubit case, the standard solution to the problem of unitary
approximation by circuits over a gate library is given by
the Solovay-Kitaev algorithm \cite{Dawson2008}.  The multiple qubit case may be handled
via employing results from \cite{Barenco1995} that show how to decompose any $n$-qubit
unitary into a circuit with CNOT and single-qubit gates.  Given precision
$\varepsilon$,
the Solovay-Kitaev algorithm produces a sequence of gates of length
$O\left(\log^{c}\left(1/\varepsilon\right)\right)$
and requires time $O\left(\log^{d}\left(1/\varepsilon\right)\right)$, for positive 
constants $c$ and $d$. 

While the Solovay-Kitaev algorithm provides a provably efficient approximation,
it does not guarantee finding an exact decomposition of the unitary into a circuit 
if there is one, nor does it answer the question of whether an exact
implementation
exists.  We refer to these as the problems of {\em exact} synthesis.  
Studying the problems related to exact synthesis is the focus of our paper.  
In particular, we study the relation between single-qubit unitaries and circuits
composed with Clifford and
T gates.  We answer two main questions: first, given a unitary how to
efficiently decide if it can be synthesized exactly, and second, how to find an efficient gate sequence
that implements a given single-qubit unitary exactly (limited to the scenario
when such an implementation exists, which we know from answering the first of
the two questions).  We further provide some intuition about the multiple qubit case.

Our motivation for this study is rooted in the observation that the
implementations of quantum algorithms exhibit errors from multiple sources, including (1) algorithmic
errors resulting from the mathematical probability of measuring a correct answer being less
than one for many quantum algorithms \cite{bk:nc}, (2) errors due to
decoherence \cite{bk:nc}, (3) systematic errors and imperfections in controlling
apparatus (e.g., \cite{ar:clj}), and (4) errors arising from the inability to
implement a desired transformation exactly using the available finite gate set requiring one to resort to
approximations.  Minimizing the effect of errors has direct implications on the resources needed to implement an algorithm and sometimes determines the very ability to implement a quantum
algorithm and demonstrate it experimentally on available hardware of a specific size.  We set out to study the fourth
type of error, rule those out whenever possible, and identify situations when
such approximation errors cannot be avoided.  During the course of this study we have
also identified that we can prove certain tight and constructive upper bounds on the circuit size
for those unitaries that may be implemented exactly.  In particular, we report a single-qubit circuit 
synthesis algorithm that guarantees optimality of both Hadamard and T gate counts.

The remainder of the paper is organized as follows.  In the next section, 
we summarize and discuss our main results.  Follow up sections contain necessary
proofs. 
In Section~\ref{sec:2}, we reduce the problem of single-qubit unitary synthesis
to the problem of state preparation. 
In Section~\ref{sec:3}, we discuss two major technical Lemmas
required to prove our main result summarized in Theorem~\ref{thm:main}.  We also
present an algorithm for efficient decomposition of single-qubit unitaries
in terms of Hadamard, H$:=\frac{1}{\sqrt{2}}\left(\begin{array}{cc}
1 & 1\\
1 & -1
\end{array}\right)$, and T gates.  Section~\ref{sec:5} and \nameref{app:A} flesh out
formal proofs of minor technical results
used in Section~\ref{sec:4}. \nameref{app:B} contains a proof
showing that the
number of Hadamard and T gates in the circuits produced by Algorithm
\ref{alg:Decomposition-of-unitary} is minimal.

\section{Formulation and discussion of the results} \label{sec:2}
\noindent
One of our two main results reported in this paper is the following theorem:
\begin{thm} \label{thm:main}
The set of $2\by 2$ unitaries over the ring $\R$ is equivalent to the set of
those unitaries implementable exactly as single-qubit circuits constructed
using\footnote{ Note, that gate H may be replaced with all Clifford group gates
without change to the meaning, though may help to visually bridge this
formulation with 
the formulation of the follow-up general conjecture.} H and T gates only. 
\end{thm}

The inclusion of the set
of unitaries implementable exactly 
via circuits employing H and T gates into the set of $2\by 2$ unitaries over the ring $\R$ is straightforward, since, indeed, all four
elements of each of the unitary matrices 
H and T belong to the ring $\R$, and circuit composition is equivalent to matrix
multiplication in the unitary matrix 
formalism.  Since both operations used in the standard definition of matrix
multiplication, ``+'' and ``$\times$'', applied to the ring elements, clearly do
not take us outside the ring, each circuit constructed using H and T gates computes a matrix
whose elements belong to the ring $\R$.  The inverse inclusion is more
difficult to prove.  The proof is discussed in Sections \ref{sec:3}-\ref{sec:5} and
\nameref{app:A}.

We believe the statement of the Theorem \ref{thm:main} may be extended and
generalized into the following conjecture:
\begin{cnj} \label{cnj:main}
For $n>1$, the set of $2^n\by 2^n$ unitaries over the ring $\R$ is equivalent to
the set of unitaries implementable exactly as circuits with 
Clifford and T gates built using $(n+1)$ qubits, where the last qubit, an
ancillary qubit, is set to the value $\ket{0}$ prior to the circuit computation, and is
required to 
be returned in the state $\ket{0}$ at the end of it. 
\end{cnj}

Note, that the ancillary qubit may not be used if its use is not required.  
However, we next show that the requirement to include a single ancillary qubit
is essential---if removed, the statement of Conjecture \ref{cnj:main} would have
been false.  
The necessity of this condition is tantamount to the vast difference between
single-qubit case and $n$-qubit case for $n>1$.

We wish to illustrate the necessity of the single ancilla with
the use of controlled-T gate, defined as follows: 
\[
\left(\begin{array}{cccc}
1 & 0 & 0 & 0\\
0 & 1 & 0 & 0\\
0 & 0 & 1 & 0\\
0 & 0 & 0 & \w
\end{array}\right),
\]
\noindent where $\w:=e^{2\pi i/8}$, the eighth root of unity.  The determinant
of this unitary is $\w$.  However, any Clifford
gate as well as the T gate viewed as matrices over a set of two qubits have a
determinant that is a power of the imaginary number $i$.  Using
the multiplicative property of the determinant we conclude that the circuits
over the Clifford and T library may implement only those 
unitaries whose determinant is a power of the imaginary $i$.  As such,
the controlled-T, whose determinant equals $\w$, cannot be implemented 
as a circuit with Clifford and T gates built using only two qubits. It is also impossible
to implement the controlled-T up to global phase.  The reason is that the only complex numbers of the form $e^{i \phi}$ 
that belong to the ring $\R$ are $\w^{k}$ for integer $k$, as it is shown in \nameref{app:A}.  Therefore, global phase can only change determinant by 
a multiplicative factor of $\w^{4k}$. 
However, as
reported in \cite{Amy} and illustrated in Figure \ref{fig:c-t}, 
an implementation of the controlled-T over a set of three qubits, one of which
is set to and returned in the state $\ket{0}$, exists.  With the 
addition of an ancillary qubit, as described, the determinant argument fails,
because one would now need to look 
at the determinant of a subsystem, that, unlike the whole system, may be
manipulated in such a way as to allow the computation to happen.

\begin{figure}[t!]
\centerline{\includegraphics[scale=0.37]{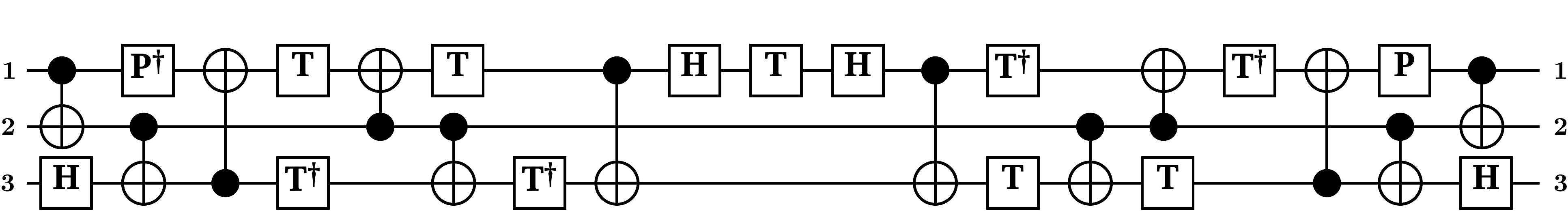}} 
\vspace*{13pt}
\caption{\label{fig:c-t} Circuit implementing the controlled-T gate, with upper
qubit being the control, middle qubit being the target, 
and bottom qubit being the ancilla. Reprinted from \cite{Amy}.}
\end{figure}

Theorem \ref{thm:main} provides an easy to verify criteria
that reliably differentiates between unitaries implementable in the H and T 
library and those requiring approximation.  As an example, $R_x(\frac{\pi}{3})$
and gates such as $R_z(\frac{\pi}{2^m})$, where $m>3$, popular in the
construction of 
circuits for the Quantum Fourier Transform (QFT), cannot be implemented exactly and must be approximated.  Thus,
the error in approximations may be an unavoidable feature for certain quantum
computations.  
Furthermore, Conjecture \ref{cnj:main}, whose one inclusion is trivial---all Clifford and T circuits compute unitaries over the ring
$\R$---implies that 
the QFT over more than three qubits may not be computed exactly as a circuit with
Clifford and T gates, and must be approximated.  

Our second major result is an algorithm (Algorithm \ref{alg:Decomposition-of-unitary}) that
synthesizes a quantum single-qubit circuit using gates H, Z$:=$T$^4$,
P$:=$T$^2$, and T in time $O\left(n_{opt}\right)$, where $n_{opt}$ is the
minimal number of gates required to implement a given unitary.  Technically, the
above complexity calculation assumes that the operations over the ring $\R$ take
a fixed finite amount of time.  In terms of bit operations, however, this time is 
quadratic in $n_{opt}$.  Nevertheless, assuming ring operations take constant time, the efficiency has a surprising
implication.  In particular, it is easy to show that our algorithm is
asymptotically optimal, in terms of both its speed and quality guarantees, among
all algorithms (whether known or not) solving the problem of synthesis in the
single-qubit case.  Indeed, a natural lower bound to accomplish the task of
synthesizing a unitary is $n_{opt}$---the minimal time it takes to simply write
down an optimal circuit assuming a certain algorithm somehow knows what it
actually is. Our 
algorithm features the upper bound of $O\left(n_{opt}\right)$ matching the lower
bound and implying asymptotic optimality.  To state the above somewhat
differently, the problem in approximating a unitary by a circuit is that of
finding an approximating unitary with elements in the ring $\R$, but not composing the 
circuit itself. We formally show H- and T-optimality of the circuits synthesized by Algorithm \ref{alg:Decomposition-of-unitary} in \nameref{app:B}.

The T-optimality of circuit decompositions has been a topic of study of the recent paper~\cite{Sv}. 
We note that our algorithm guarantees both T- and H-optimality, whereas 
the one reported in \cite{Sv} guarantees only T-optimality.  Furthermore, our implementation allows a trade-off between 
the number of Phase and Pauli-Z gates (the number of other gates used, being Pauli-X and Pauli-Y, does not exceed a total of three).  We shared our 
software implementation and circuits obtained from it to facilitate proper comparison of the two synthesis algorithms.

In the recent literature, similar topics have also been studied in~\cite{Amy}
who concentrated on finding depth-optimal multiple qubit quantum circuits in the
Clifford and T library, \cite{MA} who developed a normal form for single-qubit
quantum circuits using gates H, P, and T, and \cite{Dawson2008, F:PhD} who
considered improvements of the Solovay-Kitaev algorithm that are very relevant
to our work.  In fact, we employ the Solovay-Kitaev algorithm as a tool to find an
approximating unitary that we can then synthesize using our algorithm for exact
single-qubit unitary synthesis. 

\section{Reducing unitary implementation to state preparation} \label{sec:3}
In this section we discuss the connection between state preparation
and implementation of a unitary by a quantum circuit.  In the next section, we prove the following result: 
\begin{lem}
\label{thm:implementability}Any single-qubit state with entries in the ring
$\R$ can be prepared using only H and T gates given the initial
state $\ket 0$. 
\end{lem}
We first establish why Lemma \ref{thm:implementability} implies that any single-qubit unitary
with entries in the ring $\R$ can be implemented exactly using H and T gates.

Observe that any single-qubit unitary can be
written in the form
\[
\left(\begin{array}{cc}
z & -w^{*}e^{i\phi}\\
w & z^{*}e^{i\phi}
\end{array}\right),
\]
where ${}^{*}$ denotes the complex conjugate.  The determinant of
the unitary is equal to $e^{i\phi}$ and belongs to the ring $\R$ when
all entries of the unitary belong to the ring $\R$.  It turns out that the
only elements in the ring with the absolute value of $1$ are $\w^{k}$
for integer $k$.  We postpone the proof; it follows from techniques
developed in \nameref{app:A} and discussed at the end of the
appendix.  For now, we conclude that the most general form of a unitary with entries in
the ring is: 
\[
\left(\begin{array}{cc}
z & -w^{*}\w^{k}\\
w & z^{*}\w^{k}
\end{array}\right).
\]

We next show how to find a circuit that implements any such unitary
when we know a circuit that prepares its first column given the 
state $\ket 0$.  Suppose we have a circuit that prepares state
$\left(\begin{array}{c}
z\\
w
\end{array}\right)$. This means that the first column of a unitary corresponding
to the circuit is $\left(\begin{array}{c}
z\\
w
\end{array}\right)$ and there exists an integer $k'$ such that the unitary is equal
to: 
\[
\left(\begin{array}{cc}
z & -w^{*}\w^{k'}\\
w & z^{*}\w^{k'}
\end{array}\right).
\]
We can synthesize all possible unitaries with the first column $\left(z,w\right)^{t}$
by multiplying the unitary above by a power of T from the right:
\[
\left(\begin{array}{cc}
z & -w^{*}\w^{k'}\\
w & z^{*}\w^{k'}
\end{array}\right)T^{k-k'}=\left(\begin{array}{cc}
z & -w^{*}\w^{k}\\
w & z^{*}\w^{k}
\end{array}\right).
\]
This also shows that given a circuit for state preparation of length $n$
we can always find a circuit for unitary implementation of length $n+O(1)$
and vice versa. 

\section{\label{sec:Sequence-for-state}Sequence for state preparation }
\label{sec:4}
\noindent
We start with an example that illustrates the main ideas needed to
prove Lemma \ref{thm:implementability}.  Next we formulate two
results, Lemma \ref{thm:Main} and Lemma \ref{thm:the-second}, that the proof of Lemma \ref{thm:implementability} is based on. 
Afterwards, we describe the algorithm for decomposition of a
unitary with entries in the ring $\R$ into a sequence of H and T
gates.  Finally, we prove Lemma \ref{thm:Main}.  The proof of Lemma \ref{thm:the-second}
is more involved and it is shown in Section \ref{sec:Bilinear-forms-and}.

Let us consider a sequence of states $\left(HT\right)^{n}\ket 0$.
It is an infinite sequence, since in the Bloch sphere picture unitary
$HT$ corresponds to rotation over an angle that is an irrational
fraction of $\pi$.  Table \ref{tab:First-5-elements} shows the first
four elements of the sequence. 

\begin{table}
\caption{\label{tab:First-5-elements}First four elements of sequence
$\left(HT\right)^{n}\ket 0$}
\centerline{\footnotesize
\begin{tabular}{l|l|l}
\hline 
$n$ & $\tbl{\left(HT\right)^{n}\ket 0=\left(\begin{array}{c}
z_{n}^{\,}\\
w_{n}^{\,}
\end{array}\right)}$ & $\tbl{\left(\begin{array}{c}
\left|z_{n}\right|{}^{2}\\
\left|w_{n}\right|{}^{2}
\end{array}\right)}$\tabularnewline
\hline 
\hline 
$1$ & $\tbl{\frac{1}{\sqrt{2}}\left(\begin{array}{c}
1\\
1
\end{array}\right)}$ &
$\tbl{\frac{1}{\left(\sqrt{2}\right)^{2}}\left(\begin{array}{c}
1\\
1
\end{array}\right)}$\tabularnewline
\hline 
$2$ & $\tbl{\frac{1}{\left(\sqrt{2}\right)^{2}}\left(\begin{array}{c}
\omega+1\\
1-\w
\end{array}\right)}$ &
$\tbl{\frac{1}{\left(\sqrt{2}\right)^{3}}\left(\begin{array}{c}
\sqrt{2}+1\\
\sqrt{2}-1
\end{array}\right)}$\tabularnewline
\hline 
$3$ & $\tbl{\frac{1}{\left(\sqrt{2}\right)^{2}}\left(\begin{array}{c}
\omega^{2}-\omega^{3}+1\\
\omega^{\,}
\end{array}\right)}$ &
$\tbl{\frac{1}{\left(\sqrt{2}\right)^{4}}\left(\begin{array}{c}
3\\
1
\end{array}\right)}$\tabularnewline
\hline 
$4$ & $\tbl{\frac{1}{\left(\sqrt{2}\right)^{3}}\left(\begin{array}{c}
2\omega^{2}-\omega^{3}+1\\
1-\omega^{3}
\end{array}\right)}$ &
$\tbl{\frac{1}{\left(\sqrt{2}\right)^{5}}\left(\begin{array}{c}
3\sqrt{2}-1\\
\sqrt{2}+1
\end{array}\right)}$\tabularnewline
\hline 
\end{tabular}}

\end{table}

There are two features in this example that are important. First
is that the power of $\sqrt{2}$ in the denominator of the entries is
the same. We prove that the power of the denominator is the same
in the general case of a unit vector with entries in ring $\R$. The second
feature is that the power of $\sqrt{2}$ in the denominator of
$\left|z_{n}\right|{}^{2}$
increases by $1$ after multiplication by $HT$. We show that
in general, under additional assumptions, multiplication by
$H\left(T^{k}\right)$
cannot change the power of $\sqrt{2}$ in the denominator by more than $1$. Importantly, under the
same additional assumptions it is always possible to find such an integer
$k$ that the power increases or decreases by $1$. 

We need to clarify what we mean by power of $\sqrt{2}$ in the denominator,
because, for example, it is possible to write $\frac{1}{\sqrt{2}}$
as $\frac{\w-\w^{3}}{2}$.  As such, it may seem that the power of $\sqrt{2}$
in the denominator of a number from the ring $\R$ is not well defined.
To address this issue we consider the subring 
\[
\Zr := \left\{ a+b\w+c\w^{2}+d\w^{3}, a,b,c,d \in \Z \right\}
\]
of ring $\R$ and the smallest denominator exponent. These definitions are also crucial
for our proofs. 

It is natural to extend the notion of divisibility to elements of $\Zr$:
$x$ divides $y$ when there exists $x'$ from the ring $\Zr$
such that $xx'=y$. Using the divisibility relation we can introduce
the smallest denominator exponent and greatest dividing exponent.
\begin{defn}
The {\em smallest denominator exponent}, $\sde\left(z,x\right)$, of base
$x\in\Zr$ with respect to $z\in\R$ is the smallest integer value $k$ such that $zx^{k}\in\Zr$.
If there is no such $k$, the smallest denominator exponent is infinite. 
\end{defn}
For example, $\sde(1/4,\sqrt{2})=4$ and 
$\sde\left(2\sqrt{2},\sqrt{2}\right)=-3$.
The smallest denominator exponent of base $\sqrt{2}$ is finite
for all elements of the ring $\R$.  The greatest dividing exponent is closely
connected to $\sde$. 
\begin{defn}
The {\em greatest dividing exponent}, $\gde\left(z,x\right)$, of base $x\in\Zr$
with respect to $z\in\Zr$ is the integer value $k$ such that
$x^{k}$ divides $z$ and $x$ does not divide quotient $\frac{z}{x^k}$.
If no such $k$ exists, the greatest dividing exponent is
said to be infinite. 
\end{defn}
For example, $\gde\left(z,\w^{n}\right)=\infty$, since $\w^{n}$
divides any element of $\Zr$, and $\gde\left(0,x\right)=\infty$.
For any non-zero base $x\in\Zr$, $\gde$ and $\sde$ are related via a simple formula:
\begin{equation}
\sde\left(\frac{z}{x^{k}},x\right)=k-\gde\left(z,x\right)\label{eq:gdesde}.
\end{equation}
This follows from the definitions of $\sde$ and $\gde$. First, the
assumption $\gde\left(z,x\right)=k_{0}$ implies
$\sde\left(\frac{z}{x^{k}},x\right)\ge k-k_{0}$.
Second, the assumption $\sde\left(\frac{z}{x^{k}},x\right)=k_{0}$
implies $\gde\left(z,x\right)\ge k+k_{0}$.  Since both inequalities need to be satisfied 
simultaneously, this implies the equality. 

We are now ready to introduce
two results that describe the change of the $\sde$ as a result of the application
$H\left(T\right)^{k}$ to a state: 

\[
HT^{k}\left(\begin{array}{c}
z\\
w
\end{array}\right)=\left(\begin{array}{c}
\frac{z+w\w^{k}}{\sqrt{2}}\\
\frac{z-w\w^{k}}{\sqrt{2}}
\end{array}\right).
\]

\begin{lem}
\label{thm:Main}Let $\left(\begin{array}{c}
z\\
w
\end{array}\right)$ be a state with entries in $\R$ and let $\sde\left(\abs
z^{2}\right)\ge 4$.
Then, for any integer $k$: 
\begin{equation}
-1\le\sde\left(\abs{\frac{z+w\w^{k}}{\sqrt{2}}}^{2}\right)-\sde\left(\abs
z^{2}\right)\le 1\label{eq:main}.
\end{equation}

\end{lem}
The next lemma states that for almost all unit vectors the difference
in (\ref{eq:main}) achieves all possible values, when the power of $\w$
is chosen appropriately. 
\begin{lem}
\label{thm:the-second}Let $\left(\begin{array}{c}
z\\
w
\end{array}\right)$ be a state with entries in $\R$ and let $\sde\left(\abs
z^{2}\right)\ge 4$.
Then, for each number $s \in \{-1,0,1\}$ there exists an integer $k \in\left\{ 0,1,2,3\right\}$
such that: 
\[
\sde\left(\abs{\frac{z+w\w^{k}}{\sqrt{2}}}^{2}\right)-\sde\left(\abs
z^{2}\right)=s.
\]
\end{lem}
These lemmas are essential for showing how to find a sequence of gates that prepares
a state with entries in the ring $\R$ given the initial state $\ket 0$.
Now we sketch a proof of Lemma \ref{thm:implementability}. Later,
in Lemma \ref{lem:sde}, we show that for arbitrary $u$ and $v$ from the
ring $\R$ the equality $|u|^{2}+|v|^{2}=1$ implies $\sde(|u|^{2})=\sde(|v|^{2})$,
when $\sde\left(\abs u^{2}\right)\ge 1$ and $\sde\left(\abs v^{2}\right)\ge 1$.
Therefore, under assumptions of Lemma \ref{thm:Main}, we may consider
$\sde$ of a single entry in any given state.  Lemma \ref{thm:the-second} implies
that we can prepare any state using H and T gates if we can prepare
any state $\left(\begin{array}{c}
z\\
w
\end{array}\right)$ such that $\sde(|z|^{2})\le 3$.  The set of states with
$\sde(|z|^{2})\le 3$ is finite and small.  Therefore, we can exhaustively verify that all
such states can be prepared using H and T gates given the initial
state $\ket 0$.  In fact, we performed such verification using a breadth first
search algorithm. 

The statement of Lemma \ref{thm:the-second} remains true if we replace the set
$\left\{ 0,1,2,3\right\} $ by $\left\{ 0,-1,-2,-3\right\}$.  Lemma \ref{thm:the-second}
results in Algorithm \ref{alg:Decomposition-of-unitary}
for decomposition of a unitary matrix with entries in the ring $\R$ into
a sequence of H and T gates.  Its complexity is
$O\left(\sde(|z|^{2})\right)$,
where $z$ is an entry of the unitary. The idea behind the algorithm
is as follows: given a $2\by 2$ unitary $U$ over the ring $\R$ and $\sde\ge 4$,
there is a value of $k$ in $\left\{ 0,1,2,3\right\} $ such that the
multiplication by $H\left(T^{k}\right)$ reduces the $\sde$ by
$1$. Thus, after $n-4$ steps, we have expressed 
\[
U=HT^{k_{1}}H\ldots HT^{k_{n-4}}U',
\]
where any entry $z'$ of $U'$ has the property $\sde\left(\abs{z'}^{2}\right)<4$.
The number of such unitaries is small enough to handle the decomposition
of $U'$ via employing a breadth-first search algorithm.

We use $n_{opt}(U)$ to define the smallest length of the circuit that implements $U$. 

\begin{cor}
Algorithm \ref{alg:Decomposition-of-unitary} produces circuit of length $O(n_{opt}(U))$ 
and uses $O(n_{opt}(U))$ arithmetic operations.  The number of bit operations it uses is $O(n_{opt}^2(U))$. \end{cor}
\begin{proof}
Lemma \ref{lem:sde}, proved later in this section,
implies that the value of $\sde\left(\abs{\cdot}^{2}\right)$ is the same for all
entries of $U$ when the $\sde$ of at least one entry is greater
than $0$. For such unitaries we define
$\sde^{\abs{\cdot}^{2}}\left(U\right)=\sde\left(\abs{z'}^{2}\right)$,
where $z'$ is an entry of $U$.  The remaining special case is unitaries
of the form 
\[
\left(\begin{array}{cc}
0 & \w^{k}\\
\w^{j} & 0
\end{array}\right),\left(\begin{array}{cc}
\w^{k} & 0\\
0 & \w^{j}
\end{array}\right).
\]
We define $\sde^{\abs{\cdot}^{2}}$ to be $0$ for all of them. Consider
a set $S_{opt,3}$ of optimal H and T circuits for unitaries with
$\sde^{\abs{\cdot}^{2}}\le 3$.
This is a finite set and therefore we can define $N_{3}$ to be
the maximal number of gates in a circuit from $S_{opt,3}$.  If we have a circuit
that is optimal and its length is greater than $N_{3}$, the corresponding
unitary must have $\sde^{\abs{\cdot}^{2}}\ge 4$. Consider now a unitary
$U$ with an optimal circuit of length $n_g\left(U\right)$ that is larger
than $N_{3}$.  As it is optimal, all its subsequences are optimal
and it does not include H${}^2$.  Let $N_{H,3}$ be the maximum of the number
of Hadamard gates used by the circuits in $S_{opt,3}$. An optimal circuit for $U$
includes at most $\left\lfloor \frac{n_g\left(U\right)-N_{3}}{2}\right\rfloor
+N_{H,3}$
Hadamard gates and, by Lemma \ref{thm:Main}, $\sde^{\abs{\cdot}^{2}}$
of the resulting unitary is less or equal to $N_{H,3}+3+\left\lfloor
\frac{n_g\left(U\right)-N_{3}}{2}\right\rfloor $.
We conclude that for all unitaries except a finite set: 
\[
\sde^{\abs{\cdot}^{2}}\left(U\right)\le N_{H,3}+3+\left\lfloor
\frac{n_g\left(U\right)-N_{3}}{2}\right\rfloor.
\]
From the other side, the decomposition algorithm we described gives
us the bound
\[
n_g\left(U\right)\le N_{3}+4\cdot\sde^{\abs{\cdot}^{2}}\left(U\right).
\]
We conclude that $n_g\left(U\right)$ and
$\sde^{\abs{\cdot}^{2}}\left(U\right)$
are asymptotically equivalent. Therefore the algorithm's runtime is
$O\left(n_g\left(U\right)\right)$,
because the algorithm performs $\sde^{\abs{\cdot}^{2}}\left(U\right)-4$
steps.
    
We note that to store $U$ we need $O\left(\sde^{\abs{\cdot}^{2}}\left(U\right)\right)$ bits 
and therefore the addition on each step of the algorithm requires  $O\left(\sde^{\abs{\cdot}^{2}}\left(U\right)\right)$ bit 
operations. Therefore we use $O(n_{opt}^2(U))$ bit operations in total.
\end{proof}

This proof illustrates the technique that we use in \nameref{app:B} to
find a tighter connection between $\sde$ and the circuit implementation cost, in particular we prove that circuits produced by the algorithm are H- and T-optimal.

\begin{algorithm}[H]
\begin{algorithmic}
\Require  Unitary $U=\left(\begin{array}{cc}
z_{00} & z_{01}\\
z_{10} & z_{11}
\end{array}\right)$ with entries in the ring $\R$. \\
$\mathbb{S}_3$ -- table of all unitaries over the ring $\R$,
such that $sde$ of their entries is less than or equal to 3.
\Ensure Sequence $S_{out}$ of H and T gates that implements $U$.

\State $S_{out}\leftarrow Empty$
\State $s\leftarrow \sde(|z_{00}|^{2})$

\While{s$>$3}
\State state$\leftarrow$unfound
\ForAll{$k\in\{0,1,2,3\}$}
	\While{ state = unfound }
	\State $z'_{00}\leftarrow$ top left entry of $HT^{-k}U$
		\If{$sde\left(|z_{00}'|^{2}\right)=s-1$}
			\State state = found 
			\State add $T^{k}H$ to the end of $S_{out}$
			\State $s\leftarrow \sde\left(|z_{00}'|^{2}\right)$
			\State $U\leftarrow HT^{-k}U$
		\EndIf
	\EndWhile
\EndFor
\EndWhile
\State lookup sequence $S_{rem}$ for $U$ in $\mathbb{S}_3$
\State add $S_{rem}$ to the end of $S_{out}$ \\
\Return $S_{out}$
\end{algorithmic} 

\caption{\label{alg:Decomposition-of-unitary}Decomposition of a unitary matrix
with entries in the ring $\R$.}
\end{algorithm}

We next prove Lemma \ref{thm:Main}.  In Section \ref{sec:Bilinear-forms-and}
we use Lemma \ref{thm:Main} to show that we can prove Lemma \ref{thm:the-second}
by considering a large, but finite, number of different cases. We
provide an algorithm (Algorithm \ref{alg:Verification-of-second}) that verifies all these cases.

We now proceed to the proof of Lemma \ref{thm:Main}. We use equation
(\ref{eq:gdesde}) connecting $\sde$ and $\gde$ together with the following 
properties of $\gde$. For any base $x\in\Zr$: 
\begin{align}
\gde\left(y+y',x\right)\ge\min\left(\gde\left(y,x\right),\gde\left(y',
x\right)\right)\quad\label{eq:splitting}\\
\gde\left(yx^{k},x\right)=k+\gde\left(y,x\right)\quad & \mbox{(base
extraction)}\label{eq:base-extr}\\
\gde\left(y,x\right)<\gde\left(y',x\right)\Rightarrow\gde\left(y+y',
x\right)=\gde\left(y,x\right)\quad & \mbox{(absorption)}\label{eq:absorption}.
\end{align}
It is also helpful to note that $\gde\left(y,x\right)$ is invariant
with respect to multiplication by $\w$ and complex conjugation of
both $x$ and $y$. 

All these properties follow directly from the definition of $\gde$;
the first three are briefly discussed in \nameref{app:A}. The
condition
$\gde\left(y,x\right)<\gde\left(y',x\right)$ is necessary for the
third property. For example,
$\gde\left(\sqrt{2}+\sqrt{2},\sqrt{2}\right)\ne\gde\left(\sqrt{2},\sqrt{2}
\right)$. 

There are also important properties specific to base $\sqrt{2}$.
We use shorthand $\gde\left(\cdot\right)$ for
$\gde\left(\cdot,\sqrt{2}\right)$:
\begin{alignat}{1}
 & \gde\left(x\right)=\gde\left(\abs x^{2},2\right)\label{eq:gde2s2}\\
 & 0\le\gde\left(\abs x^{2}\right)-2\gde\left(x\right)\le 1\label{eq:gde2s2in}\\
 & \gde\left(Re\left(\sqrt{2}xy^{*}\right)\right)\ge\left\lfloor
\frac{1}{2}\left(\gde\left(\abs x^{2}\right)+\gde\left(\abs
y^{2}\right)\right)\right\rfloor \label{eq:gdestr}\\
 & \gde\left(\abs x^{2}\right)=\gde\left(\abs
y^{2}\right)\Rightarrow\gde\left(x\right)=\gde\left(y\right)\label{eq:reduction}.
\end{alignat}

Proofs of these properties are not difficult but tedious; furthermore, for completeness they are included
in \nameref{app:A}. We exemplify them here. In the second property,
inequality (\ref{eq:gde2s2in}), when $x=\w$ the left inequality becomes
equality and for $\w+1$ the right one does.  When we substitute $x=\w,y=\w+1$
in the second to last property, inequality (\ref{eq:gdestr}), it turns into $0=\left\lfloor
\frac{1}{2}\right\rfloor $,
so the floor function $r\mapsto\left\lfloor r\right\rfloor $
is necessary. For the third property it is important that
$\re\left(\sqrt{2}xy^{*}\right)$
is an element of $\Zr$ when $x,y$ itself belongs to the ring $\Zr$.
In contrast, $\re\left(xy^{*}\right)$ is not always an element of $\Zr$,
in particular, when $x=\w,y=\w+1$. In general,
$\gde\left(x\right)=\gde\left(y\right)$
does not imply $\gde\left(\abs x^{2}\right)=\gde\left(\abs y^{2}\right)$.
For instance, $\gde\left(\w+1\right)=\gde\left(\w\right)$, but
$\abs{\w+1}^{2}=2+\sqrt{2}$
and $\abs{\w}^{2}=1$.

In the proof of Lemma \ref{thm:Main} we use
$x=z\left(\sqrt{2}\right)^{\sde\left(z\right)},y=w\left(\sqrt{2}\right)^{
\sde\left(w\right)}$
that are elements of $\Zr$. The next lemma shows an additional
property that such $x$ and $y$ have.
\begin{lem}
\label{lem:sde}Let $z$ and $w$ be elements of the ring $\R$ such that
$\abs z^{2}+\abs w^{2}=1$ and $\sde\left(z\right)\ge 1$ or
$\sde\left(w\right)\ge 1$,
then $\sde\left(z\right)=\sde\left(w\right)$ and for elements  $x=z\left(\sqrt{2}\right)^{\sde\left(z\right)}$ and 
$y=w\left(\sqrt{2}\right)^{\sde\left(w\right)}$ of the ring $\Zr$
it holds that $\gde\left(\abs x^{2}\right)=\gde\left(\abs
y^{2}\right)\le 1$.\end{lem}
\begin{proof}
Without loss of generality, suppose $\sde\left(z\right)\ge\sde\left(w\right)$.
Using the relation in equation (\ref{eq:gdesde}) between $\sde$
and $\gde$, expressing $z$ and $w$ in terms of $x$ and $y$, and substituting
the result into equation $\abs z^{2}+\abs w^{2}=1$, we obtain 
\[
\abs
y^{2}\left(\sqrt{2}\right)^{2\left(\sde\left(z\right)-\sde\left(w\right)\right)}
=\left(\sqrt{2}\right)^{2\sde\left(z\right)}-\abs x^{2}.
\]
Substituting $z=x/\left(\sqrt{2}\right)^{\sde\left(z\right)}$ into
formula (\ref{eq:gdesde}) relating $\sde$ and $\gde$, we obtain
$\gde\left(x\right)=0$, and using one of the inequalities (\ref{eq:gde2s2in})
connecting $\gde\left(\abs x^{2}\right)$ and $\gde\left(x\right)$
we conclude that $\gde\left(\abs x^{2}\right)\le 1$.  Similarly,
$\gde\left(\abs y^{2}\right)\le 1$. We use the absorption property
(\ref{eq:absorption}) of $\gde\left(\cdot\right)$ to write: 
\[
\gde\left(\abs
y^{2}\left(\sqrt{2}\right)^{2\left(\sde\left(z\right)-\sde\left(w\right)\right)}
\right)=\gde\left(\abs x^{2}\right).
\]
Equivalently, using the base extraction property (\ref{eq:base-extr}):
\[
\gde\left(\abs
y^{2}\right)+2\left(\sde\left(z\right)-\sde\left(w\right)\right)=\gde\left(\abs
x^{2}\right).
\]
Taking into account $\gde\left(\abs x^{2}\right)\le 1$ and $\gde\left(\abs
y^{2}\right)\le 1$,
it follows that $\sde\left(z\right)=\sde\left(w\right)$.
\end{proof}

In the proof of Lemma \ref{thm:Main} we turn inequality (\ref{eq:main})
for difference of $\sde$ into an inequality for difference of $\gde\left(\abs
x^{2}\right)$
and $\gde\left(\abs{x+y}^{2}\right)$.  The following lemma shows a basic relation
between these quantities that we will use.
\begin{lem}
\label{lem:ineq}If $x$ and $y$ are elements of the ring $\Zr$ such that $\abs x^{2}+\abs
y^{2}=\left(\sqrt{2}\right)^{m}$,
then 
\[
\gde\left(\abs{x+y}^{2}\right)\ge\min\left(m,1+\left\lfloor
\frac{1}{2}\left(\gde\left(\abs x^{2}\right)+\gde\left(\abs
y^{2}\right)\right)\right\rfloor \right).
\]
\end{lem}
\begin{proof}
The first step is to expand $\abs{x+y}^{2}$ as $\abs x^{2}+\abs
y^{2}+\sqrt{2}Re\left(\sqrt{2}xy^{\ast}\right)$.
Next, we apply inequality (\ref{eq:splitting}) to the $\gde$ of the sum,
and then the base extraction property (\ref{eq:base-extr}) of the $\gde$.
We use equality $\gde\left(\abs x^{2}+\abs y^{2}\right)=m$ to conclude that
\[
\gde\left(\abs{x+y}^{2}\right)\ge\min\left(m,1+\gde\left(Re\left(\sqrt{2}xy^{\ast}
\right)\right)\right).
\]
Finally, we use inequality (\ref{eq:gdestr}) for
$\gde\left(Re\left(\sqrt{2}xy^{\ast}\right)\right)$
to derive the statement of the lemma.
\end{proof}

Now we collected all tools required to prove Lemma \ref{thm:Main}.

\begin{proof}
Recall that, we are proving that for elements $z$ and $w$ of the ring $\R$ and any
integer $k$ it is true that: 
\[
-1\le\sde\left(\abs{\frac{z+w\w^{k}}{\sqrt{2}}}^{2}\right)-\sde\left(\abs
z^{2}\right)\le 1,\mbox{ when }\sde\left(\abs z^{2}\right)\ge 4.
\]
Using Lemma \ref{lem:sde} we can define
$m=\sde\left(z\right)=\sde\left(w\w^{k}\right)$
and $x=\w^{k}z\left(\sqrt{2}\right)^{m}$ and
$y=w\left(\sqrt{2}\right)^{m}$.
Using the relation (\ref{eq:gdesde}) between $\gde$ and $\sde$, and
the base extraction property (\ref{eq:base-extr}) of $\gde$ we rewrite
the inequality we are trying to prove as: 
\[
1\le\gde\left(\abs{x+y}^{2}\right)-\gde\left(\abs x^{2}\right)\le 3.
\]
It follows from Lemma \ref{lem:sde} that $\gde\left(\abs
x^{2}\right)=\gde\left(\abs y^{2}\right)\le 1$.
Taking into account $\abs x^{2}+\abs y^{2}=\sqrt{2}^{2m}$ and applying
the inequality proved in Lemma \ref{lem:ineq} to $x$ and $y$ we conclude
that: 
\[
\gde\left(\abs{x+y}^{2}\right)\ge\min\left(2m,1+\gde\left(\abs
x^{2}\right)\right).
\]
The condition $m\ge 4$ allows us to remove taking the minimum on the right hand side and replace it with $1+\gde\left(\abs x^{2}\right)$. 
This proves one of the two inequalities we are trying to show, $1\le\gde\left(\abs{x+y}^{2}\right)-\gde\left(\abs x^{2}\right)$.
To prove the second inequality, $\gde\left(\abs{x+y}^{2}\right)-\gde\left(\abs
x^{2}\right)\le 3$,
we apply Lemma \ref{lem:ineq} to the pair of elements of the ring $\Zr$, $x+y$ and $x-y$.  The conditions of the
lemma are satisfied because $\abs{x+y}^2+\abs{x-y}^2=\sqrt{2}^{2\left(m+1\right)}$.
Therefore: 
\[
\gde\left(4\abs x^{2}\right)\ge\min\left(2\left(m+1\right),1+\left\lfloor
\frac{1}{2}\left(\gde\left(\abs{x+y}^{2}\right)+\gde\left(\abs{x-y}^{2}
\right)\right)\right\rfloor \right).
\]
Using the base extraction property (\ref{eq:base-extr}), we notice
that $\gde\left(4\abs x^{2}\right)=4+\gde\left(\abs x^{2}\right)$.
It follows from $m\ge 4$ that $2\left(m+1\right)\ge 4+\gde\left(\abs
x^{2}\right)$.
As such, we can again remove the minimization and simplify the inequality
to: 
\[
3+\gde\left(\abs x^{2}\right)\ge\left\lfloor
\frac{1}{2}\left(\gde\left(\abs{x+y}^{2}\right)+\gde\left(\abs{x-y}^{2}
\right)\right)\right\rfloor .
\]

To finish the proof it suffices to show that
$\gde\left(\abs{x+y}^{2}\right)=\gde\left(\abs{x-y}^{2}\right)$.
We establish an upper bound for $\gde\left(\abs{x+y}^{2}\right)$
and use the absorption property (\ref{eq:absorption}) of $\gde$.
Using non-negativity of $\gde$ and the definition of the floor function
we get: 
\[
2\left(3+\gde\left(\abs x^{2}\right)\right)+1\ge\gde\left(\abs{x+y}^{2}\right).
\]
Since $\gde\left(\abs x^2\right) \leq 1$, $\gde\left(\abs{x+y}^{2}\right)\le 9$.  Observing that $2\left(m+1\right)>9$
we confirm that 
\[
\gde\left(\abs{x-y}^{2}\right)=\gde\left(\sqrt{2}^{2\left(m+1\right)}-\abs{x+y}^
{2}\right)=\gde\left(\abs{x+y}^{2}\right).
\]
\end{proof}

To prove Lemma \ref{thm:the-second} it suffices to show that
$\gde(\abs{x+\w^{k}y}^{2})-\gde(\abs x^{2})$ achieves all values
in the set $\left\{ 1,2,3\right\} $ as $k$ varies over all values
in the range from $0$ to $3$.  We can split this into two cases:
$\gde(\abs x^{2})=1$ and $\gde(\abs x^{2})=0$. We need to check
if $\gde(\abs{x+\w^{k}y}^{2})$ belongs to $\left\{ 1,2,3\right\} $
or $\left\{ 2,3,4\right\}$.  Therefore, it is important to describe
these conditions in terms of $x$ and $y$.  This is accomplished in the next Section.

\section{\label{sec:Bilinear-forms-and}Quadratic forms and greatest dividing
exponent} \label{sec:5}
\noindent
We first clarify why it is enough to check a finite number
of cases to prove Lemma \ref{thm:the-second}.  Recall
how the lemma can be restated in terms of the elements of the ring $\Zr$.
Next we illustrate why we can achieve a finite number of cases with a simple
example using integer numbers $\Z$. Then we show how this idea can be extended
to the elements of the ring $\Zr$ that are real (that is, with imaginary part equal to zero).  Finally, in the proof
of Lemma \ref{thm:the-second}, we identify a set of cases that we
need to check and provide an algorithm to perform it. 

As discussed at the end of the previous Section, to prove Lemma
\ref{thm:the-second}
one can consider elements $x$ and $y$ of the ring $\Zr$ such that $\abs x^{2}+\abs
y^{2}=2^{m}$
for $m\ge 4$.  We know from Lemma \ref{thm:Main} that there are three possibilities
in each of the two cases: 
\begin{itemize}
\item when $\gde(\abs x^{2})=0$, $\gde(\abs{x+\w^{k}y}^{2})$ equals to
$1,2\mbox{, or }3$,
\item when $\gde(\abs x^{2})=1$, $\gde(\abs{x+\w^{k}y}^{2})$ equals $2,3\mbox{, 
or }4$.
\end{itemize}
We want to show that each of these possibilities is achievable for a specific choice of 
$k\in\left\{ 0,1,2,3\right\}$. 

We illustrate the idea of the reduction to a finite number of cases
with an example.  Suppose we want to describe two classes of integer numbers: 
\begin{itemize}
\item integer $a$ such that the $\gde\left(a^{2},2\right)=2$,
\item integer $a$ such that the $\gde\left(a^{2},2\right)>2$.
\end{itemize}
It is enough to know $a^{2}\, mod\,2^{3}$ to decide which class $a$
belongs to. Therefore we can consider $8$ residues $a\, mod\,2^{3}$
and find the classes to which they belong. We extend this
idea to the real elements of the ring $\Zr$, being elements of the ring $\Zr$
that are equal to their own real part.  Afterwards we apply the result
to $\abs{x+\w^{k}y}^{2}$, that is a real element of $\Zr$. 

We note that the real elements of $\Zr$ are of the form $a+\sqrt{2}b$
where $a$ and $b$ are themselves integer numbers.  An important preliminary observation, that
follows from the irrationality of $\sqrt{2}$, is that for any integer number
$c$ 
\begin{equation}
\gde\left(c\right)=2\gde\left(c,2\right).\label{eq:int}
\end{equation}
The next proposition gives a condition equivalent to
$\gde\left(a+\sqrt{2}b\right)=k$,
expressed in terms of $\gde\left(a,2\right)$ and $\gde\left(b,2\right)$: 
\begin{prop}
\label{prop:gdeeee}Let $a$ and \textbf{$b$} be integer numbers. There
are two possibilities:
\begin{itemize}
\item $\gde\left(a+\sqrt{2}b\right)$ is even if and only if
$\gde\left(b,2\right)\ge\gde\left(a,2\right)$;
in this case, $\gde\left(a,2\right)=\gde\left(a+\sqrt{2}b\right)/2$. 
\item $\gde\left(a+\sqrt{2}b\right)$ is odd if and only if
$\gde\left(b,2\right)<\gde\left(a,2\right)$;
in this case,
$\gde\left(b,2\right)=\left(\gde\left(a+\sqrt{2}b\right)-1\right)/2$.
\end{itemize}
\end{prop}
\begin{proof}
Consider the case when $\gde\left(b,2\right)<\gde\left(a,2\right)$.
Observing, from equation (\ref{eq:int}), that $\gde\left(a\right)$ is always even,
$\gde\left(a\right)>\gde\left(\sqrt{2}b\right)$,
and by the absorption property (\ref{eq:absorption}) of $\gde$ we
have $\gde\left(a+\sqrt{2}b\right)=\gde\left(\sqrt{2}b\right)$. Using
the base extraction property (\ref{eq:base-extr}) of $\gde$ and the 
relation (\ref{eq:int}) between $\gde\left(\cdot\right)$ and
$\gde\left(\cdot,2\right)$
for integers we obtain $\gde\left(a+\sqrt{2}b\right)=1+2\gde\left(b,2\right)$.
The other case similarly implies
$\gde\left(a+\sqrt{2}b\right)=2\gde\left(a,2\right)$.
In terms of real elements of the ring $\Zr$, this results in the following
relations: 
\begin{align*}
A_{1}=\left\{ \gde\left(b,2\right)<\gde\left(a,2\right)\right\}  & \subseteq
B_{1}=\left\{ \gde\left(a+\sqrt{2}b\right)\mbox{ is even}\right\}, \\
A_{2}=\left\{ \gde\left(b,2\right)\ge\gde\left(a,2\right)\right\}  & \subseteq
B_{2}=\left\{ \gde\left(a+\sqrt{2}b\right)\mbox{ is odd}\right\}. 
\end{align*}
We note that each pair of sets $\{A_{1},A_{2}\}$ and $\{B_{1},B_{2}\}$ defines
a partition of real elements of the ring $\Zr$. This completes the proof since when for partitions $\{A_{1},A_{2}\}$ and $\{B_{1},B_{2}\}$
of some set the inclusions $A_{1}\subseteq B_{1},A_{2}\subseteq B_{2}$
imply $A_{1}=B_{1}$ and $A_{2}=B_{2}$.
\end{proof}

To express $\abs{x+\w^{k}y}^{2}$ in the form $a+\sqrt{2}b$ in a concise
way, we introduce two quadratic forms $\P{\cdot}$ and
$\Q{\cdot}$ with the property: 
\begin{equation}
\abs x^{2}=\P{x}+\sqrt{2} \Q{x}.\label{eq:abs}
\end{equation}
Given that $x$, an element of $\Zr$, can be expressed
in terms of the integer number coordinates as follows, $x=x_{0}+x_{1}\w+x_{2}\w^{2}+x_{3}\w^{3}$,
we define the quadratic forms as: 
\begin{equation}
\P{x}:=x_{0}^{2}+x_{1}^{2}+x_{2}^{2}+x_{3}^{2},\label{eq:ipxx}
\end{equation}
\begin{equation}
\Q{x}:=x_{0}\left(x_{1}-x_{3}\right)+x_{2}\left(x_{1}+x_{3}\right).\label{eq:ips2xx}
\end{equation}

Let us rewrite equality $\gde\left(\abs{x+y}^{2}\right)=4$
in terms of these quadratic forms and the $\gde$ of base $2$. Using
Proposition \ref{prop:gdeeee} we can write: 
\[
\gde\left(\P{x+\w^{k}y},2\right)=2,
\]
\[
\gde\left(\Q{x+\w^{k}y},2\right)\ge 2.
\]

Similar to the example given at the beginning of this section, we see that
it suffices to know the values of the quadratic forms modulo $2^{3}$.
To compute them, it suffices to know the values of the integer coefficients
of $x$ and $y$ modulo $2^{3}$.  This follows from the expression
of the product $\w y$ in terms of the integer number coefficients: 
\[
\w\left(y_{1}+y_{2}\w+y_{3}\w^{2}+y_{4}\w^{3}\right)=-y_{4}+y_{1}\w+y_{2}\w^{2}
+y_{3}\w^{3},
\]
and from the following two observations:
\begin{itemize}
\item integer number coefficients of the sum of two elements of the ring $\R$ is the sum of their integer number
coefficients, 
\item for any element of $\Zr$, $x$, the values of quadratic forms $\P{x}$ and $\Q{x}$
modulo $2^{3}$ are defined by the values modulo $2^{3}$ of the integer number coefficients of $x$. 
\end{itemize}
In summary, to check the second part of Lemma \ref{thm:Main} we
need to consider all possible values for the integer coefficients
of $x$ and $y$ modulo $2^{3}$.  There are two additional constraints on
them. The first one is $\abs x^{2}+\abs y^{2}=2^{m}$. Since we assumed
$m\ge 4$, we can write necessary conditions to satisfy this constraint,
in terms of the quadratic forms, as: 
\[
\P{x}\equiv -\P{y} \left(mod\,2^{3}\right),
\]
\[
\Q{x} \equiv -\Q{y} \left(mod\,2^{3}\right).
\]
The second constraint is $\gde\left(\abs x^{2}\right)=\gde\left(\abs y^{2}\right)$
and $\gde\left(\abs x^{2}\right)\le 1$. To check it, we use the same
approach as in the example with $\gde\left(\abs{x+y}^{2}\right)=4$.

We have now introduced the necessary notions required to prove Lemma \ref{thm:the-second}.

\begin{proof}
Our proof is an exhaustive verification, assisted by a computer search.  We rewrite the statement of the lemma formally as follows: 
\begin{equation}
\left.\begin{array}{l}
\mathcal{G}_{j}=\left\{ \begin{array}{c|c}
\left(x,y\right)\in\Zr\times\Zr & \exists m\ge 4 \mbox{ s.t.} \abs x^{2}+\abs
y^{2}=2^{m},\\
 & \gde\left(x\right)=\gde\left(y\right)=j
\end{array}\right\} ,j\in\left\{ 0,1\right\} ,\\
\mbox{for all }\left(x,y\right)\in\mathcal{G}_{j},\mbox{for all }s\in\left\{
1,2,3\right\} \mbox{ there exists }k\in\left\{ 0,1,2,3\right\} \\
\mbox{such that }\gde\left(\abs{x+\w^{k}y}^{2}\right)=s+j.
\end{array}\right\} \label{eq:Formal}
\end{equation}
The sets $\mathcal{G}_{j}$ are infinite, so it is impossible to perform
the check directly.  As we illustrated with an example, equality
$\gde\left(\abs{x+\w^{k}y}^{2}\right)=s+j$
depends only on the values of the integer coordinates of $x$ and $y$ modulo
$2^{3}$. If the sets $\mathcal{G}_{j}$ were also defined in terms
of the residues modulo $2^{3}$ we could just check the lemma in terms
of equivalence classes corresponding to different residuals.  More
precisely, the equivalence relation $\sim$ we would use is: 
\[
\sum_{p=0}^{3}x_{p}\w^{p}\sim\sum_{p=0}^{3}y_{p}\w^{p}\overset{def}{
\Longleftrightarrow}\mbox{ for all }p\in\left\{ 0,1,2,3\right\}: x_{p}\equiv y_{p}\mt3.
\]
To address the issue, we introduce sets $\mathcal{Q}_{j}$ that include
$\mathcal{G}_{j}$ as subsets: 
\[
\mathcal{Q}_{j}=\left\{ \begin{array}{c|c}
\left(x,y\right)\in\Zr\times\Zr & \gde\left(x\right)=\gde\left(y\right)=j\\
 & \P{x}+\P{y}\equiv 0\mt3\\
 & \Q{x} + \Q{y} \equiv 0\mt3
\end{array}\right\}, j\in\left\{ 0,1\right\}. 
\]
Therefore, in terms of the equivalence classes with respect to the above defined relation
$\sim$ the more general problem can be verified in a finite number
of steps.  However, the number of equivalence classes is large.  This is why 
we employ a computer search that performs verification of all cases.  To rewrite (\ref{eq:Formal})
into conditions in terms of the equivalence classes it suffices to replace
$\mathcal{G}_{j}$ by $\mathcal{Q}_{j}$, replace $x$ and $y$ by their
equivalence classes, and replace $\Zr$ by the set of equivalence classes $\Zr/\sim$.

Algorithm \ref{alg:Verification-of-second} verifies Lemma \ref{thm:the-second}.
We use bar (e.g., $\overline{x}$ and $\overline{y}$) to represent 4-dimensional vectors with entries in $\Z_{8}$, the ring of residues
modulo $8$.  The definition of bilinear forms, multiplication by $\w$
and the relations $\gde\left(\abs{\cdot}^{2}\right)=1,2,3,4$ extend to
$\overline{x}$ and $\overline{y}$.  We implemented Algorithm \ref{alg:Verification-of-second} 
and the result of its execution is $true$.  This completes the proof.
\end{proof}

\begin{algorithm}[H]
\begin{algorithmic}
\Ensure Returns $true$ if the statement of Lemma \ref{thm:the-second} is correct; otherwise, returns $false$. \\
\Comment Here, $G_{j,a,b}$ is the set of all residue vectors $\ox$ such that $\gde(\ox)=j,\P{\ox}=a,\Q{\ox} =b$.
\ForAll { $x_1,x_2,x_3,x_4 \in \{0,\ldots,7\}$ } \Comment generate possible
residue vectors;
	\State $\ox\leftarrow(x_1,x_2,x_3,x_4)$
	\State $j\leftarrow \gde(\abs{\ox}^{2})$, $a\leftarrow
\P{\ox}, b\leftarrow \Q{\ox}$
	\If{ $j \in \{0,1\}$ }
		\State add $\overline{x}$ to $G_{j,a,b}$ 
	\EndIf
\EndFor
\ForAll{ $j\in\{0,1\},a_x \in \{0,7\},b_x \in \{0,7\} $ }
	\State $a_y \leftarrow -a_x\,mod\,8$, $b_y \leftarrow -b_x\,mod\,8$ 
\Comment consider only those pairs that
	\ForAll{ $(\ox,\oy) \in G_{j,a_x,b_x}\times G_{j,a_y, b_y }$ } \Comment
satisfy necessary conditions;
		\ForAll{ $d \in \{1,2,3\}$ } 
			\State state $\leftarrow$ unfound
			\ForAll{ $k \in \{0,1,2,3\}$ }
				\State $\overline{t} \leftarrow \ox + \w^k\oy$
				\If{ $\gde(\abs{\overline{t}}^{2})$ = d + j }
					\State state $\leftarrow$ found
				\EndIf
			\EndFor
			\If {state = unfound} 
				\State \Return $false$
			\EndIf
		\EndFor
	\EndFor
\EndFor
\State \Return $true$
\end{algorithmic}
\caption{\label{alg:Verification-of-second}Verification of Lemma \ref{thm:the-second}.}
\end{algorithm}

\section{Implementation}
Our C++ implementation of Algorithm \ref{alg:Decomposition-of-unitary} is available online at \url{http://code.google.com/p/sqct/}.

\section{Experimental results} 
\noindent
Table \ref{tab:1} summarizes the results of first obtaining an
approximation of the given rotation matrix by a unitary over the ring $\R$
using our implementation of the Solovay-Kitaev algorithm \cite{Dawson2008, K1},
and then decomposing it into a circuit using the exact synthesis Algorithm
\ref{alg:Decomposition-of-unitary} presented in this paper.  We note that the 
implementation of our synthesis Algorithm \ref{alg:Decomposition-of-unitary} (runtimes found in the column $t_{decomp}$) 
is significantly faster than the implementation of the Solovay-Kitaev algorithm used to approximate the unitary (runtimes reported in the column $t_{approx}$). 
Furthermore, we were able to calculate approximating 
circuits using 5 to 7 iterations of the Solovay-Kitaev algorithm followed by our synthesis algorithm.  The total runtime to approximate and decompose unitaries
ranged from approximately 11 to 600 seconds, correspondingly, featuring best approximating errors on the order of $10^{-50}$, 
and circuits with up to millions of gates.  Actual specifications of all circuits reported, as well as those synthesized but not explicitly 
included in the Table \ref{tab:1}, due to space constraints, may be obtained from \url{http://qcirc.iqc.uwaterloo.ca/}. 

On each step Algorithm \ref{alg:Decomposition-of-unitary} chooses from one of four small circuits: H, HT, HT$^2$ (=HP), and HT$^3$ to reduce $\sde$.  In 
practice, Pauli-Z gate is often easier to implement than either Phase or T gate.  The cost of P$^{\dagger}$ and T$^{\dagger}$ gates is usually the same as that of the respective P and T gates.  
We took this into account by writing circuit HT$^3$ using an equivalent and cheaper form HZT$^{\dagger}$.  This significantly reduces the number of Phase gates required to implement a unitary. 
If there is no preference between the choice of P or Z, or P is preferred to Z, the HPT could be used in place of HT$^3$.

\renewcommand{\arraystretch}{1.3}
\begin{table}
\caption[Results of the approximation of $\Rot_z(\varphi)$]{\label{tab:1}Results of the approximation of 
$\Rot_z(\varphi)=\left(
\begin{array}{cc}
e^{-i \varphi} & 0 \\
0 & e^{i \varphi} 
\end{array} \right)$ by our implementation.  
Column $N_{I}$ contains the number of iterations used by the Solovay-Kitaev algorithm, $n_g$---total
number of gates (sum of the next four columns), $n_{T}$---number of T and T$^{\dagger}$ 
gates, $n_{H}$---number of Hadamard gates, $n_{P}$---number of P and P$^{\dagger}$ 
gates, $n_{Pl}$---number of Pauli gates (note that the combined number of Pauli-X and Pauli-Y is never more than three for any of the circuits, so 
$n_{Pl}$ is dominated by Pauli-Z gates), $dist$---trace distance to
approximation, $t_{approx}$---time spent on the unitary approximation using the Solovay-Kitaev algorithm (in seconds),
$t_{decomp}$---time spent on the decomposition of the approximating unitary into circuit, 
per Algorithm \ref{alg:Decomposition-of-unitary} (in seconds). Circuit specifications are available at \url{http://qcirc.iqc.uwaterloo.ca/}.
}
\centerline{\footnotesize 
\begin{tabular}{|c||c||c|c|c|c|c||c||c|c|}
\hline 
$U$ & $N_{I}$ & $n_g$ & $n_{T}$ & $n_{H}$ & $n_{P}$ & $n_{Pl}$ & $dist$ & $t_{approx}$ & $t_{decomp}$\tabularnewline
\hline 
\hline 
\multirow{5}{*}{$R_{z}\left(\frac{\pi}{16}\right)$} 
& 0 & 74 & 28 & 27 & 2 & 17 & $1.34296\times 10^{-3}$ & 0.09180 & 0.00022\tabularnewline \cline{2-10}
& 1 & 342 & 132 & 132 & 1 & 77 & $4.61204\times 10^{-5}$ & 0.12238 & 0.00079\tabularnewline \cline{2-10}
& 2 & 1683 & 670 & 670 & 1 & 342 & $5.68176\times 10^{-7}$ & 0.22369 & 0.00391\tabularnewline \cline{2-10}
& 3 & 8197 & 3284 & 3283 & 0 & 1630 & $2.97644\times 10^{-10}$ & 0.83023 & 0.02077\tabularnewline \cline{2-10}
& 4 & 35819 & 14312 & 14311 & 3 & 7193 & $3.64068\times 10^{-15}$ & 2.91120 & 0.12453\tabularnewline \cline{2-10}
\hline 
\hline 
\multirow{5}{*}{$R_{z}\left(\frac{\pi}{32}\right)$} 
& 0 & 64 & 24 & 23 & 2 & 15 & $3.92540\times 10^{-4}$ & 0.01805 & 0.00021\tabularnewline \cline{2-10}
& 1 & 314 & 124 & 123 & 2 & 65 & $1.34267\times 10^{-5}$ & 0.05238 & 0.00074\tabularnewline \cline{2-10}
& 2 & 1388 & 556 & 556 & 0 & 276 & $4.65743\times 10^{-7}$ & 0.31252 & 0.00321\tabularnewline \cline{2-10}
& 3 & 7493 & 3000 & 2999 & 1 & 1493 & $1.10252\times 10^{-10}$ & 0.90493 & 0.01950\tabularnewline \cline{2-10}
& 4 & 35113 & 14054 & 14053 & 2 & 7004 & $2.69806\times 10^{-15}$ & 2.96710 & 0.12122\tabularnewline \cline{2-10}
\hline 
\hline 
\multirow{5}{*}{$R_{z}\left(\frac{\pi}{64}\right)$} 
& 0 & 54 & 22 & 23 & 2 & 7 & $8.05585\times 10^{-4}$ & 0.08827 & 0.00020\tabularnewline \cline{2-10}
& 1 & 344 & 136 & 136 & 3 & 69 & $9.57729\times 10^{-6}$ & 0.12163 & 0.00080\tabularnewline \cline{2-10}
& 2 & 1414 & 564 & 564 & 3 & 283 & $1.97877\times 10^{-7}$ & 0.38928 & 0.00326\tabularnewline \cline{2-10}
& 3 & 7769 & 3086 & 3087 & 3 & 1593 & $1.08884\times 10^{-10}$ & 0.98522 & 0.01954\tabularnewline \cline{2-10}
& 4 & 35456 & 14170 & 14171 & 2 & 7113 & $3.00267\times 10^{-15}$ & 3.05440 & 0.12384\tabularnewline \cline{2-10}
\hline 
\hline 
\multirow{5}{*}{$R_{z}\left(\frac{\pi}{128}\right)$} 
& 0 & 72 & 28 & 29 & 1 & 14 & $9.59916\times 10^{-4}$ & 0.08822 & 0.00023\tabularnewline \cline{2-10}
& 1 & 344 & 136 & 137 & 3 & 68 & $1.79353\times 10^{-5}$ & 0.12143 & 0.00081\tabularnewline \cline{2-10}
& 2 & 1588 & 634 & 634 & 4 & 316 & $3.67734\times 10^{-7}$ & 0.39048 & 0.00368\tabularnewline \cline{2-10}
& 3 & 7519 & 3004 & 3005 & 2 & 1508 & $4.23657\times 10^{-10}$ & 0.98045 & 0.01890\tabularnewline \cline{2-10}
& 4 & 34388 & 13722 & 13722 & 2 & 6942 & $1.32046\times 10^{-14}$ & 2.86740 & 0.11832\tabularnewline \cline{2-10}
\hline 
\hline 
\multirow{5}{*}{$R_{z}\left(\frac{\pi}{256}\right)$} 
& 0 & 71 & 28 & 29 & 2 & 12 & $5.06207\times 10^{-4}$ & 0.01819 & 0.00023\tabularnewline \cline{2-10}
& 1 & 326 & 136 & 136 & 2 & 52 & $1.08919\times 10^{-5}$ & 0.05474 & 0.00079\tabularnewline \cline{2-10}
& 2 & 1389 & 566 & 567 & 3 & 253 & $2.00138\times 10^{-7}$ & 0.30498 & 0.00332\tabularnewline \cline{2-10}
& 3 & 7900 & 3174 & 3175 & 3 & 1548 & $2.91716\times 10^{-10}$ & 0.91405 & 0.02060\tabularnewline \cline{2-10}
& 4 & 38188 & 15290 & 15291 & 1 & 7606 & $8.87785\times 10^{-15}$ & 2.98030 & 0.13545\tabularnewline \cline{2-10}
\hline 
\hline 
\multirow{5}{*}{$R_{z}\left(\frac{\pi}{512}\right)$} 
& 0 & 76 & 30 & 29 & 2 & 15 & $3.62591\times 10^{-4}$ & 0.01749 & 0.00023\tabularnewline \cline{2-10}
& 1 & 319 & 126 & 126 & 2 & 65 & $1.95491\times 10^{-5}$ & 0.05171 & 0.00075\tabularnewline \cline{2-10}
& 2 & 1722 & 680 & 680 & 2 & 360 & $2.76529\times 10^{-7}$ & 0.30618 & 0.00396\tabularnewline \cline{2-10}
& 3 & 8122 & 3242 & 3242 & 2 & 1636 & $1.87476\times 10^{-10}$ & 0.92576 & 0.02109\tabularnewline \cline{2-10}
& 4 & 34974 & 13992 & 13992 & 1 & 6989 & $5.66762\times 10^{-15}$ & 3.16060 & 0.11920\tabularnewline \cline{2-10}
\hline 
\hline 
\multirow{5}{*}{$R_{z}\left(\frac{\pi}{1024}\right)$} 
& 0 & 0 & 0 & 0 & 0 & 0 & $2.16938\times 10^{-3}$ & 0.08622 & 0.00005\tabularnewline \cline{2-10}
& 1 & 264 & 106 & 105 & 2 & 51 & $5.57373\times 10^{-5}$ & 0.13615 & 0.00063\tabularnewline \cline{2-10}
& 2 & 1541 & 622 & 622 & 3 & 294 & $1.74595\times 10^{-7}$ & 0.23445 & 0.00366\tabularnewline \cline{2-10}
& 3 & 6791 & 2722 & 2722 & 1 & 1346 & $5.39912\times 10^{-11}$ & 0.82811 & 0.01703\tabularnewline \cline{2-10}
& 4 & 32983 & 13188 & 13188 & 1 & 6606 & $5.54995\times 10^{-16}$ & 2.98480 & 0.11494\tabularnewline \cline{2-10}
\hline 
\end{tabular}}
\end{table}

\begin{figure}[t]
\begin{center}
\includegraphics{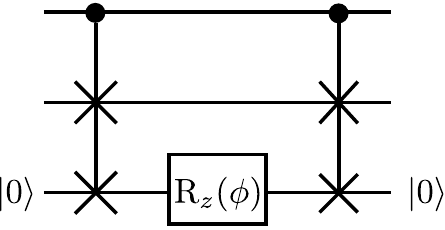}
\end{center}
\vspace{13pt}
\caption{\label{fig:c-r} Circuit implementing the controlled-$\Rot_z(\phi)$ gate. Upper
qubit is the control, middle qubit is the target, 
and bottom qubit is the ancilla.}
\end{figure}

The RAM memory requirement during unitary approximation stage for our implementation is 2.1GB. In our experiments we used a single core 
of the Intel Core i7-2600 (3.40GHz) processor.

The experimental results reported may be utilized in the construction of an approximate implementation of the 
Quantum Fourier Transform (QFT).  Figure \ref{fig:c-r} shows a circuit that
employs the technique from \cite{swapTr} to implement the controlled-$\Rot_z(\phi)$ using a single ancillary qubit.  Note that $\ket{0}$ 
is the eigenvector of $\Rot_z(\phi)$ with the eigenvalue $1$, thus this construction works correctly; moreover, no phase is introduced.  Such controlled 
rotations are used in the standard implementation of the QFT \cite{bk:nc}.  The advantage of such a circuit is that it
introduces only a small additive constant overhead on the number of gates required to turn an uncontrolled rotation into a controlled rotation. 
Indeed, the number of T gates (those are more complex in the fault tolerant implementations 
than the Clifford gates \cite{F:PhD}) required for the exact implementation of the Fredkin gate, being 7 \cite{Amy}, is small compared 
to the number of T gates required (the T-counts in the circuits we synthesize are provably minimal for the unitary being implemented) in the approximations of the individual single-qubit rotations, Table \ref{tab:1}.
In comparison to other approaches to constructing a controlled gate out of an uncontrolled gate, such as replacing each gate in the circuit 
by its controlled version, or using the decomposition provided by Lemma 5.1 in \cite{Barenco1995} (achieving a two-qubit controlled gate using 
three single-qubit uncontrolled gates, that has the expected effect of roughly tripling the number of T gates), the proposed approach 
where a single ancilla is employed appears to be beneficial. 

We approximate the controlled-$\Rot_z(\phi)$ by replacing $\Rot_z(\phi)$ with its approximation $\Rot_z'(\phi)$.
To evaluate the quality of such approximation we need to take 
into account that ancillary qubit is always initialized to $\ket{0}$ and that the
controlled rotation is a part of a larger circuit.  For this reason, 
we computed the completely bounded trace norm \cite{cbn} of the difference of the channels
corresponding to the controlled-$\Psi_{\Rot_z(\phi)}$ and its approximation using $\Psi_{ \Rot_z'(\phi) }$.  Both channels 
map the space of two qubit density matrices into the space of three qubit density matrices, as one of the inputs is fixed.

To compute the completely bounded trace norm we used the semidefinite program (SDP)
described in \cite{cbn}.  Usual 64-bit machine precision was not enough for our purposes, so
we used package SDPA-GMP \cite{sdpa} that employs The GNU Multiple Precision Arithmetic Library to solve SDP.
Also, Mathematica was used to generate files with the description of SDP problems in the input format
of SDPA-GMP.  We found that for all unitaries in Table \ref{tab:1} the ratio of completely bounded trace norm of 
$\Psi_{\Rot_z(\phi)} - \Psi_{ \Rot_z'(\phi) }$ and trace distance between $\Rot_z(\frac{\pi}{2^n})$ and its approximation belongs 
to the interval $[2.82842,2.98741]$.  In other words, the numerical value of the approximation error for the controlled rotations using the trace norm is roughly three 
times that for the corresponding single-qubit rotations using the trace distance (as per Table \ref{tab:1}).  

\begin{figure}[t]
\centerline{\includegraphics[width=1.0\textwidth]{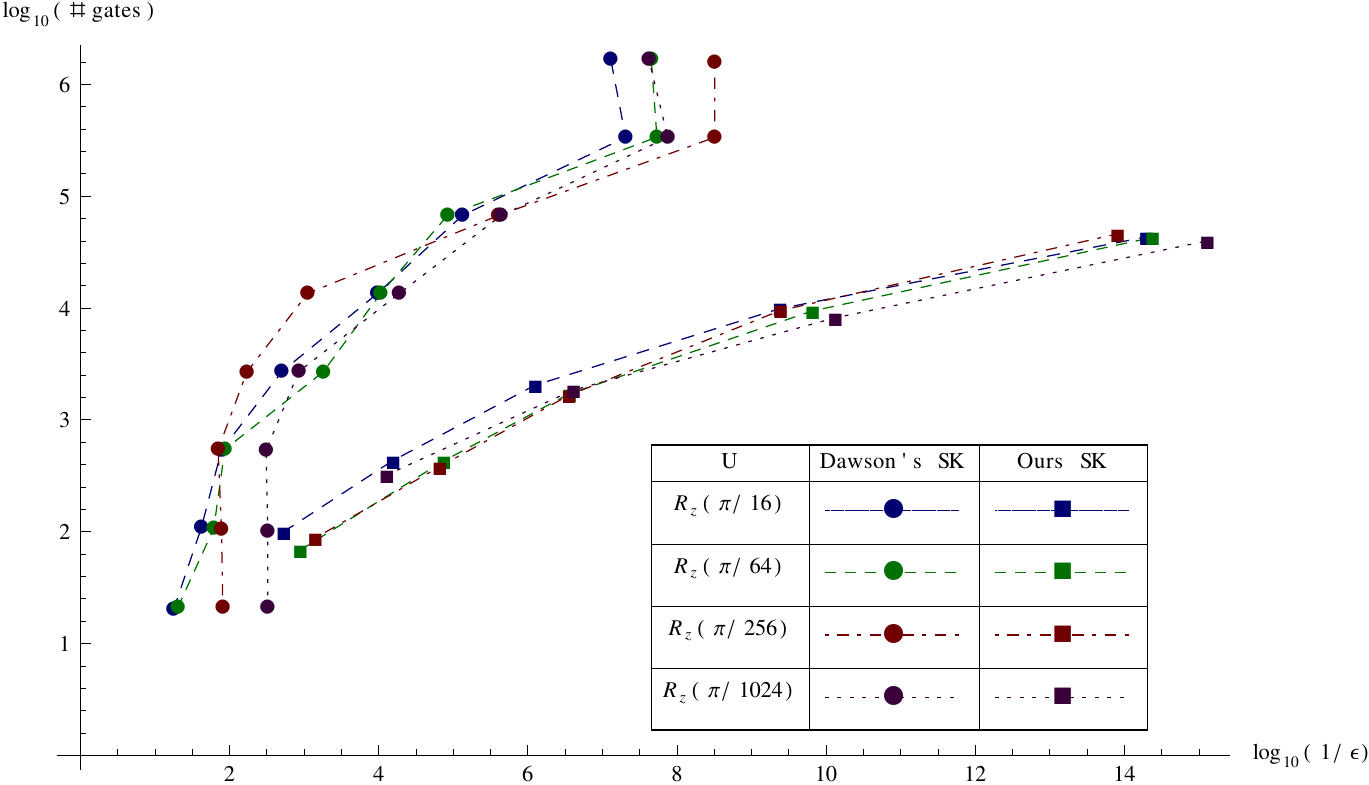}} 
\vspace*{13pt}
\caption{\label{fig:comp}Comparison between ours and Dawson's implementations of the Solovay-Kitaev algorithm. Vertical axis shows $\log_{10}$ of the number of gates
and horizontal axis shows $\log_{10}\left(1/\epsilon\right)$, where
$\epsilon$ is the projective trace distance between unitary and its
approximation.}
\end{figure}

We also performed a comparison to Dawson's implementation of the Solovay-Kitaev
algorithm (see Figure~\ref{fig:comp}) available at \url{http://gitorious.org/quantum-compiler/}. We ran Dawson's code
using the gate library \{H, T\} with the maximal sequence length equal to 22 and tile width equal to 0.14. 
During this experiment, the memory usage was around 6 GB.  For the purpose of the comparison gate counts 
for our implementation are also provided in the \{H, T\} library (Z=T${}^4$ and P=T${}^2$).  We used projective trace distance to measure quality of 
approximation as it is the one used in Dawson's code.  Because of the larger epsilon net used in our implementation we were 
able to achieve better approximation quality using fewer iterations of the Solovay-Kitaev
algorithm.  Usage of The GNU Multiple Precision Arithmetic Library allowed us to achieve precision up to $10^{-50}$ while 
Dawson's code encounters convergence problem when precision reaches $10^{-8}$.  The latter explains behaviour of the last set of points in the
experimental results for Dawson's code reported in Figure \ref{fig:comp}. 

Two other experiments that we performed with Dawson's code are a resynthesis of the
circuits generated by it using our exact decomposition algorithm.  We first
resynthesised circuits that were generated by Dawson's code using the \{H,T\}
library.  In most cases, the gate counts reduced by about 10-20\% (our resulting
circuits were further decomposed such as to use the \{H, T\} gate library).  In the other experiment we used
\{H, T, P, Z\} gate library with Dawson's implementation.  In this case, we were able
to run Dawson's code with the sequences of length 9 only, and it used 6 GB of memory.  The
gate counts, using our algorithm, decreased by about 40-60\%.  

\section{Conclusion}
\noindent
In this paper, we studied quantum circuits over the Clifford and T library. We proved that in the single-qubit case the set of unitaries over the ring $\R$ is equivalent to the set of 
unitaries computable by the Clifford and T circuits.  We generalized this statement to conjecture that in the $n$-qubit 
case the sets of unitaries computable by the Clifford and T circuits and those over the ring $\R$ are equivalent, as long 
as a single ancillary qubit residing in the state $\ket{0}$ is supplied.  While we did not prove this conjecture, we showed 
the necessity of ancilla. 

We have also presented a single-qubit synthesis algorithm that uses Pauli, H, P, and T gates.  Our algorithm is asymptotically optimal 
in both its performance guarantee, and its complexity.  The algorithm generates circuits with a provably minimal number of Hadamard and T gates. 
Furthermore, our experiments suggest that the P-counts may also be minimal.  The total number of times Pauli-X and Pauli-Y gates 
are used in any given circuit generated by the algorithm is limited to at most three.  As such, our algorithm is likely optimal (up to, possibly, a small additive constant on some of the gate counts, that, in turn, 
may be corrected by re-synthesizing the lookup table using a different/suitable circuit cost metric) in all parameters, except, possibly, Pauli-Z count.

\section*{Acknowledgements}
\noindent
We wish to thank Martin Roetteler for helpful discussions. 

Authors supported in part by the Intelligence Advanced Research Projects
Activity (IARPA) via Department of Interior National Business Center
Contract number DllPC20l66. The U.S. Government is authorized to reproduce
and distribute reprints for Governmental purposes notwithstanding
any copyright annotation thereon. Disclaimer: The views and conclusions
contained herein are those of the authors and should not be interpreted
as necessarily representing the official policies or endorsements,
either expressed or implied, of IARPA, DoI/NBC or the U.S. Government.

This material is based upon work partially supported by the National Science Foundation (NSF), during D. Maslov's assignment at the Foundation.
Any opinion, findings, and conclusions or recommendations expressed in this material are those of the author(s) and do not necessarily reflect the views of the National Science Foundation.

Michele Mosca is also supported by Canada's NSERC, MITACS, CIFAR, and CFI.
IQC and Perimeter Institute are supported in part by the Government of Canada
and the Province of Ontario.


\section*{Appendix~A}
\label{app:A}
\noindent
\label{app:1}
Here we prove properties of the greatest dividing exponent that was defined
and used in Section \ref{sec:Sequence-for-state}.  We first discuss
the base extraction property (\ref{eq:base-extr}) of $\gde$ and then
proceed to the proof of special properties of $\gde\left(\cdot,\sqrt{2}\right)$.
The base extraction property simplifies proofs of all statements related
to $\gde\left(\cdot,\sqrt{2}\right)$. 
\begin{prop}[Base extraction property]
If $x,y\in\Zr$, then for any non negative integer number $k$
\[
\gde\left(yx^{k},x\right)=k+\gde\left(y,x\right).
\]
\end{prop}
\begin{proof}
Follows directly from the definition of $\gde$.
\end{proof}

The base extraction property together with non-negativity
of $\gde$ provide a simple formula to lower bound the value of $\gde$: if $x^{k}$ divides
$y$ then $\gde\left(y,x\right)\ge k$.  Inequality for $\gde$ of a
sum (\ref{eq:splitting}) follows directly from this---$x^{\min\left(\gde\left(y,x\right),\gde\left(y',x\right)\right)}$
divides $y+y'$. The proof of absorption property (\ref{eq:absorption}) follows easily, as well. 

Now we prove properties of $\gde$ specific to base $\sqrt{2}$. Instead
of proving them for all elements of $\Zr$ it suffices to prove
them for elements of $\Zr$ that are not divisible by $\sqrt{2}$. We illustrate
this with an example $\gde\left(x,\sqrt{2}\right)=\gde\left(\abs x^{2},2\right)$.
We can always write $x=x'\left(\sqrt{2}\right)^{\gde\left(x\right)}$.
By the definition of $\gde$, $\sqrt{2}$ does not divide $x'$. By substituting
the expression for $x$ into $\gde\left(\abs x^{2},2\right)$ and then using the
base extraction property we get: 
\[
\gde\left(\abs
x^{2},2\right)=\gde\left(\abs{x'}^{2},2\right)+\gde\left(x,\sqrt{2}\right).
\]
Therefore, it suffices to show that $\gde\left(\abs{x'}^{2},2\right)=0$
when $\sqrt{2}$ does not divide $x'$, or, equivalently, when
$\gde\left(x'\right)=0$. 

The quadratic forms defined in Section \ref{sec:Bilinear-forms-and} will
be a useful tool for later proofs. Bilinear forms that generalize
them are important for the proof of relation for
$\gde\left(\re\left(xy^{\ast}\right)\right)$.
Effectively, we only need the values of mentioned forms modulo $2$. For
this reason, we also introduce forms that are equivalent modulo $2$
and more convenient for the proofs. 

We define function $\F{\cdot}{\cdot}$ for $x,y\in\Zr$ as follows: 
\[
\F xy:=x_{0}y_{0}+x_{1}y_{1}+x_{2}y_{2}+x_{3}y_{3}.
\]

Note that the following equality holds, and provides some intuition behind the choice to introduce $\F{\cdot}{\cdot}$:
\[
\re\left(xy^{\ast}\right)=\F xy+\frac{1}{\sqrt{2}}\F{\sqrt{2}x}y.
\]

Using formula $\sqrt{2}=\w-\w^{3}$ we can rewrite multiplication by $\sqrt{2}$
as a linear operator: 
\begin{equation}
\sqrt{2}x=\sqrt{2}(x): x_0+x_1\w+x_2\w^2+x_3\w^3 \mapsto \left(x_{1}-x_{3}\right)+\left(x{}_{0}+x{}_{2}\right)\w+\left(x{}_{1}
+x{}_{3}\right)\w^{2}+\left(x_{2}-x_{0}\right)\w^{3}\label{eq:sqrt2}.
\end{equation}
In particular, it is easy to verify that:  
\[
\F{\sqrt{2}x}y=\left(x{}_{1}-x{}_{3}\right)y_{0}+\left(x{}_{0}+x{}_{2}\right)y_
{1}+\left(x{}_{1}+x{}_{3}\right)y_{2}+\left(x_{2}-x_{0}\right)y_{3},
\]
and, substituting $y=x$, 
\[
\F{\sqrt{2}x}x=2\left(x_{1}-x_{3}\right)x_{2}+2\left(x{}_{1}+x{}_{3}\right)x_{0} = 2\Q{x},
\]
which corresponds to the earlier definition shown in equation (\ref{eq:ips2xx}).  The definition of $\F{\cdot}{\cdot}$ written
for $x=y$ results in an earlier definition (\ref{eq:ipxx}).  This shows how $\F{\cdot}{\cdot}$ generalizes and ties together previously 
introduced $P(\cdot)$ and $Q(\cdot)$.

Furthermore, in modulo $2$ arithmetic the following expressions hold true: 
\begin{equation}
\P{x} \equiv \left(x_{1}+x_{3}\right)+\left(x_{0}+x_{2}\right)\m\label{eq:ipxxm2}
\end{equation}
\begin{equation}
\Q{x} \equiv \left(x_{1}+x_{3}\right)\left(x_{0}+x_{2}
\right)\m\label{ip2m2}
\end{equation}
\begin{equation}
\F{\sqrt{2}x}y\equiv \left(x_{1}+x_{3}\right)\left(y_{0}+y_{2}\right)+\left(x_{0}+x_{
2}\right)\left(y_{1}+y_{3}\right)\m\label{eq:s2xy}.
\end{equation}
It is easy to verify these equations by expanding the left and right hand sides. 

The next proposition shows how we use equivalent quadratic and bilinear
forms.
\begin{prop}
\label{prop:alternatives}If $\gde\left(x\right)=0$ there are only
two alternatives:
\begin{itemize}
\item $\P x$ is even and $\Q x$ is odd,
\item $\P x$ is odd and $\Q x$ is even.
\end{itemize}
\end{prop}
\begin{proof}
The equality $\gde\left(x\right)=0$ implies that $2$ does not divide
$\sqrt{2}x$. Using expression (\ref{eq:sqrt2}) for $\sqrt{2}x$
in terms of integer coefficients we conclude that at least one of the
four numbers $x'_{1}\pm x'_{3}$, $x'_{0}\pm x'_{2}$ must be odd.
Suppose that $x'_{1}+x'_{3}$ odd. Using formulas
(\ref{eq:ipxxm2},\ref{ip2m2})
we conclude that the values of $P(x)$ and $Q(x)$ must have different parity.  The remaining three
cases are similar.
\end{proof}

An immediate corollary is: $\gde\left(x\right)=0$ implies $\gde\left(\abs
x^{2},2\right)=0$.
To show this it suffices to use expression (\ref{eq:abs}) for
$\abs x^{2}$ in terms of quadratic forms. 

We can also conclude that $\sqrt{2}$ divides $x$ if and only if
$2$ divides $\abs x^{2}$. Sufficiency follows from the definition of
$\gde$. To prove that $2$ divides $\abs x^{2}$ implies $\sqrt{2}$
divides $x$, we assume that $2$ divides $\abs x^{2}$ and $\sqrt{2}$
does not divide $x$, which leads to a contradiction. This also results in
the inequality $\gde\left(\abs x^{2}\right)\le 1$ when $\gde\left(x\right)=0$. 

We use the next two propositions to prove the inequality for
$\re\left(\sqrt{2}xy^{\ast}\right)$.
\begin{prop}
Let $\gde\left(x\right)=0$: 
\begin{itemize}
\item if $\sqrt{2}$ divides \textup{$\abs x^{2}$} then $\P x$ is even
and $\Q x$ is odd,
\item if $\sqrt{2}$ does not divide \textup{$\abs x^{2}$} then $\P x$
is odd and $\Q x$ is even.
\end{itemize}
\end{prop}
\begin{proof}
As discussed, the previous proposition implies that $\sqrt{2}$
divides $y$ if and only if $2$ divides $\abs y^{2}$. We apply this
to $\abs x^{2}$. By expressing $\abs x^{4}$ in terms of quadratic
forms we get: 
\[
\abs x^{4}=\P{x}^{2}+2\Q{x}^{2}+2\sqrt{2}\P{x}\Q{x}.
\]
We see that $2$ divides $\abs x^{4}$ if and only if $2$ divides
$\P{x}^{2}$, or, equivalently, $\sqrt{2}$ divides $\abs x^{2}$
if and only if $\P{x}$ even. Using the previous proposition again, this
time for $x$, we obtain the required result.
\end{proof}

\begin{prop}
Let $\gde\left(x\right)=0$ and $\gde\left(y\right)=0$. If $\sqrt{2}$
divides \textup{$\abs x^{2}$} and $\sqrt{2}$ divides $\abs y^{2}$
then $\sqrt{2}$ divides $\re\left(\sqrt{2}xy^{\ast}\right)$. \end{prop}
\begin{proof}
By the previous proposition, $\sqrt{2}$ divides $\abs x^{2}$ and $\sqrt{2}$
divides $\abs y^{2}$ implies that $\Q{x}$ and
$\Q{y}$ are odd. Formula (\ref{ip2m2}) implies that in terms of the integer number
coefficients of $x$ and $y$ integer numbers
$x_{1}+x_{3}$, $x_{0}+x_{2}$, $y_{1}+y_{3}$, $y_{0}+y_{2}$, are
all odd. Expressing $\re\left(\sqrt{2}xy^{\ast}\right)$ in terms
of $\F{\cdot}{\cdot}$, 
\[
\re\left(\sqrt{2}xy^{\ast}\right)=\sqrt{2}\F xy+\F{\sqrt{2}x}y,
\]
and using expression (\ref{eq:s2xy}), we conclude that $2$ divides $\F{\sqrt{2}x}y$; therefore
$\sqrt{2}$ divides $\re\left(\sqrt{2}xy^{\ast}\right)$.
\end{proof}

Now we show $\gde\left(\re\left(\sqrt{2}xy^{\ast}\right)\right)\ge\left\lfloor
\frac{1}{2}\left(\gde\left(\abs x^{2}\right)+\gde\left(\abs
y^{2}\right)\right)\right\rfloor $.
As we discussed in the beginning, we can assume $\gde\left(x\right)=0$
and $\gde\left(y\right)=0$ without loss of generality. This implies
$\gde\left(\abs x^{2}\right)\le 1$ and $\gde\left(\abs y^{2}\right)\le 1$.
The expression $\left\lfloor \frac{1}{2}\left(\gde\left(\abs
x^{2}\right)+\gde\left(\abs y^{2}\right)\right)\right\rfloor $
can only be equal to $0$ or $1$. The second one is only possible when $\gde\left(\abs
x^{2}\right)=1$ and $\gde\left(\abs y^{2}\right)=1$, in which case the previous proposition implies
$\gde\left(\re\left(\sqrt{2}xy^{\ast}\right)\right)\ge 1$. In the
first case inequality is true because of the non-negativity of $\gde$. 

We can also use quadratic forms to describe all numbers $z$ in the ring
$\R$ such that $\abs z^{2}=1$. Seeking a contradiction, suppose
$\sde\left(z\right)\ge 1$.
We can always write $z=\frac{x}{\left(\sqrt{2}\right)^{k}}$ where
$k=\sde\left(z\right)$ and $\gde\left(x\right)=0$. From the other
side $\abs x^{2}=\P{x}+\sqrt{2}\Q{x}=2^{k}$.
Thus we have a contradiction with the statement of Proposition \ref{prop:alternatives}.
We conclude that $z$ is an element of $\Zr$. Therefore we can
write $z$ in terms of its integer number coordinates, 
$z=z_{0}+z_{1}\w+z_{2}\w^{2}+z_{3}\w^{3}$.
Equality $\abs z^{2}=1$ implies that $\F
zz=z_{0}^{2}+z_{1}^{2}+z_{2}^{2}+z_{3}^{2}=1$.
Taking into account that $z_{j}$ are integer numbers we conclude that $z\in\left\{
\w^{k},k=0,\ldots,7\right\} $.

\section*{Appendix~B}
\label{app:B}

Here we prove that Algorithm \ref{alg:Decomposition-of-unitary} produces circuits with the minimal number
of Hadamard and T gates over the gate library $\G$ consisting of Hadamard, T, T${}^{\dagger}$,
P, P${}^{\dagger}$, and Pauli-X, Y, and Z gates.  We say that a circuit implements a unitary $U$ if the unitary corresponding 
to the circuit is equal to $U$ up to global phase.  We define integer-valued quantities $h(U)$ and $t(U)$ as 
the minimal number of Hadamard and T gates over all circuits implementing $U$.  We call a circuit 
H- or T-optimal if it contains the minimal number of H or T gates, correspondingly. 

\begin{thm}\label{prop:hnsde}
Let $U$ be a $2\by 2$ unitary over the ring $\R$ with a matrix entry $z$ such that $\sde(\abs z^{2})\ge4.$
Algorithm~\ref{alg:Decomposition-of-unitary} produces a circuit that implements $U$ over $\mathcal{G}$ with:
\begin{enumerate}
\item the minimal number of Hadamard gates and $h(U)=\sde(\abs z^{2})-1$, and 
\item the minimal number of T gates and $t(U)=h(U) - 1 + (l \mod 2) + (j \mod 2)$, where $l$ and $j$ are chosen such that $h(H T^l U T^j H)=h(U)+2$.
\end{enumerate}
\end{thm}
\begin{proof}

\noindent {\bf 1: H-optimality.} Using brute force, we explicitly verified that the set of H-optimal circuits with
precisely 3 Hadamard gates is equal to the set of all unitaries
over the ring $\R$ with $\sde(\abs z^{2})=4$. Suppose
we have a unitary $U$ with $\sde(\abs z^{2})=n\ge 4$. With the help of
Algorithm \ref{alg:Decomposition-of-unitary} we can reduce it to
a unitary with $\sde(\abs z^{2})=4$ while using $n-4$ Hadamard
gates to accomplish this.  As such, there exists a circuit with $n-1$ Hadamard gates that
implements $U$. 

Now consider an H-optimal circuit $C$ that implements $U$. Using brute
force, we established that if $C$ has less than 3 Hadamard gates, then $\sde(\abs
z^{2})$ is less than $4$. Suppose $C$ contains $m\ge 3$ Hadamard gates. Its
prefix, containing $3$ Hadamard gates, must also be H-optimal, and
therefore $\sde(\abs z^{2})$ of the corresponding unitary is
4. Now, using the inequality from Lemma \ref{thm:Main}, we conclude
that $\sde(\abs z^{2})$ of the unitary corresponding to $C$
is less than $m+1$.  This implies $n\le m+1$.  Since we already know that
$m\le n-1$, we may conclude that $m=n-1$ and $m$ is the number
of Hadamard gates in the circuit produced by Algorithm
\ref{alg:Decomposition-of-unitary}
in combination with the brute force step.

\noindent {\bf 2: T-optimality.} To prove T-optimality we introduce a normal form for circuits over $\G$.  We call a circuit
HT-normal if there is precisely one T gate between every two H gates and, symmetrically, precisely one H gate between every two T gates.  
It is not difficult to modify Algorithm \ref{alg:Decomposition-of-unitary} to produce a circuit in HT-normal 
form while preserving its H-optimality.  To accomplish that, first, recall that HT$^3$ = HZT$^{\dagger}$ and that all circuits generated 
during the brute force stage are both H-optimal and in the HT-normal form.  Second, any circuit produced by the algorithm 
is H-optimal and does not contain a non-H-optimal (up to global phase) subcircuit HT$^2$H = HPH = $\w$PHP. 

We will show that any H-optimal circuit in HT-normal form is also T-optimal. We start with a special case of 
HT-normal circuits---those that begin and end with the Hadamard gate, in other words, those that can be written as HS$_1$H$\ldots$HS$_k$H, and are H-optimal. Let $U$ be a unitary corresponding
to this circuit. Due to HT-normality, each $S_i$ contains exactly one T gate, the number of T gates in the circuit is $k$, and $h(U)=k+1$; therefore, $t(U)\le h(U)-1$. To prove that
$t(U) = h(U)-1$, it suffices to show that $t(U) \ge h(U)-1$. Let us write a T-optimal circuit 
for $U$ as C$_0$TC$_1$T$\ldots$TC$_k$. Each subcircuit C$_k$ implements a unitary 
from the Clifford group. Each unitary from the single-qubit Clifford group can be implemented using 
at most one H gate (recall, that we are concerned with the implementations up to global phase), therefore $ h(U) \le t(U)+1$, as required. 

In the general case, consider a circuit obtained by Algorithm \ref{alg:Decomposition-of-unitary} and implementing a unitary $V$ with $h(V)\ge3$ that is H-optimal 
and written in HT-normal form and show that it is T-optimal. We
can write it as $\mathrm{S}_0\mathrm{HS}_1\mathrm{H}\ldots \mathrm{HS}_k \mathrm{HS}_{k+1}$.
By Lemma \ref{thm:the-second} we can always find such $l$ and $j$ that 
$C:=$HT$^l$S$_0$HS$_1$H$\ldots$HS$_k$H S$_{k+1}$T$^j$H is also an H-optimal circuit. Indeed,
according to Lemma \ref{thm:the-second}, using the connection between $\sde(\cdot)$ and $h(\cdot)$ described in the first part of the proof, given $h(V)=k+1$ we can always find $l$ such that $h(\mathrm{HT}^l V)=k+2$. From the other side, 
circuit HT$^l$S$_0$HS$_1$H$\ldots$HS$_k$HS$_{k+1}$ contains $k+2$ Hadamard gates and therefore is
H-optimal. We repeat the same procedure to find $j$. 

Considering the different possible values of $l$ and $j$ allows to complete the proof of the Theorem.  This is somewhat tedious, and we illustrate how to handle different cases 
with a representative example of $l=3$ and $j=2$. In such a case, we can rewrite circuit $C$ as $C'=$HTPS$_0$HS$_1$H$\ldots$HS$_k$HS$_{k+1}$PH. We conclude 
that S$_0$ must have zero T gates and S$_{k+1}$ must have one T gate. Otherwise subcircuits HTPS$_0$H and 
HS$_{k+1}$PH will not be H-optimal. As such, we reduced the problem to the special case considered above, therefore circuit $C'$ is T-optimal and S$_0$HS$_1$H$\ldots$HS$_k$HS$_{k+1}$ is T-optimal as its subcircuit. In the general case, the following formula may be developed $t\left(V\right) = h(U) - 1 + (l \mod 2) + (j \mod 2)$.
\end{proof}

\end{document}